\documentclass[journal,twoside,web]{ieeecolor}
\usepackage{silence}
\usepackage[dvipsnames]{xcolor}
\definecolor{subsectioncolor}{rgb}{0,0,0}
\usepackage{generic}
\def\BibTeX{{\rm B\kern-.05em{\sc i\kern-.025em b}\kern-.08em
    T\kern-.1667em\lower.7ex\hbox{E}\kern-.125emX}}
\usepackage{amsmath,amssymb,amsfonts}
\usepackage{algorithmic}
\usepackage{algorithm,algorithmic}
\usepackage{textcomp}
\usepackage{bbold}
\usepackage{url}
\usepackage{cite}
\allowdisplaybreaks
\usepackage[hidelinks]{hyperref}
\usepackage[noabbrev,capitalise,nameinlink]{cleveref}
\usepackage{nicefrac, xfrac}
\usepackage[draft]{changes}
\usepackage{enumerate}

\usepackage{graphicx}
\graphicspath{ {./Figures/} }
\usepackage{subcaption}

\usepackage{siunitx}
\DeclareSIUnit{\pu}{p.u.}
\DeclareSIUnit\year{yr}
\NewDocumentCommand\SIpi{}{\text{\ensuremath{\pi}}}
\usepackage{aliascnt}

\newaliascnt{proposition}{theorem}
\newtheorem{proposition}[proposition]{Proposition}
\aliascntresetthe{proposition}
\crefname{proposition}{Proposition}{Propositions}

\newaliascnt{corollary}{theorem}
\newtheorem{corollary}[corollary]{Corollary}
\aliascntresetthe{corollary}
\crefname{corollary}{Corollary}{Corollaries}

\newaliascnt{lemma}{theorem}
\newtheorem{lemma}[lemma]{Lemma}
\aliascntresetthe{lemma}
\crefname{lemma}{Lemma}{Lemmas}

\newtheorem{definition}{Definition}
\newtheorem{remark}{Remark}
\newtheorem{assumption}{Assumption}
\crefname{assumption}{Assumption}{Assumptions}

\usepackage[short]{optidef}

\usepackage{tikz}
\usetikzlibrary{positioning}
\usetikzlibrary{matrix}
\usetikzlibrary {arrows.meta}
\usetikzlibrary{shapes}
\usetikzlibrary{calc}
\usetikzlibrary{math}
\usetikzlibrary{decorations.markings}
\usetikzlibrary{angles,quotes}
\usepackage[european,straightvoltages]{circuitikz}

\usepackage{xr}
\newcommand{\lhb}[1]{\textcolor{black}{#1}} 
\newcommand{\lhc}[1]{\textcolor{black}{#1}} 
\newcommand{\lhg}[1]{\textcolor{black}{#1}} 
\newcommand{\lhh}[1]{\textcolor{black}{#1}} 
\newcommand{\lhi}[1]{\textcolor{black}{#1}} 

\title{\
Decentralised Plug-and-Play Stability Conditions for AC Grids---Part II: Unstable Subsystems}

\author{Liam Hallinan and Ioannis Lestas %
    \thanks{L. Hallinan and I. Lestas are with the Department of Engineering, University of Cambridge, Trumpington Street, Cambridge, CB2 1PZ, United Kingdom. Emails:
        {\tt\small <lh706, icl20>@cam.ac.uk}.}%
}

\begin{document}
\maketitle

\begin{abstract}
    Part I of this paper presented a decentralised framework for certifying small-signal stability in AC grids using frequency-domain quadratic constraints on individual grid subsystems. Part II extends the framework \lhg{to a broader class of networks that contain} unstable subsystems. 
    \lhh{In particular, we show that such unstable subsystems arise in many common scenarios, even when the aggregate system is stable and well behaved.}
    \lhg{Such subsystems must be stabilised by the closed-loop network interconnection, which complicates decentralised stability analysis}.
    A stable \textit{hybrid} \lhh{representation resembling the \lhi{power} ($PQ$) model at low frequencies and the \lhi{impedance} ($IV$) model at high frequencies} 
    is then defined, to which the stability framework of Part I can be applied, 
    \lhi{allowing plug-and-play compatible conditions to be formulated.}
    \lhi{Furthermore, we show that at low frequencies, the characteristic loci in the Nyquist plot of the return ratio in this hybrid representation decouple}
    into unbounded and bounded branches along the classical $P$--$\delta$ / $Q$--$V$ separation, and give sufficient conditions on each branch \lhi{for ensuring system stability}. 
    The results are validated on the Kundur two-area system, where \lhh{stability is certified, with} a decentralised, plug-and-play compatible grid code \lhh{covering} the frequency ranges in which electromagnetic interactions arise. 
\end{abstract}

\section{Introduction}
\IEEEPARstart{T}{he} increasingly distributed nature of AC power systems means that classical small-signal stability techniques --- such as eigenvalue analysis of monolithic state-space models --- are no longer scalable. In response, Part I of this paper proposed a framework that certifies stability by imposing quadratic constraints on the frequency-domain input--output properties of local subsystems, which collectively imply the multivariable Nyquist criterion for the entire grid.
The framework unifies many prior results and offers increased generality, including conditions on extended subsystems that enable plug-and-play functionality.

The case study in Part I demonstrated this reduced conservatism, but exposed two limitations. First, despite the grid being stable in closed loop, two of the bus impedances contained unstable poles, meaning the bus admittance representation had to be used for decentralised analysis. As the corresponding network system is not sparse \lhi{when an admittance representation is used}, this representation does not admit conditions on extended subsystems, forfeiting the plug-and-play benefits they offer. Second, decentralised conditions were infeasible at low frequencies unless a stiff voltage source was included in the network, so a centralised check was required in this region.

Part II of this paper addresses both limitations. We first present an equivalent power representation of the grid, here termed the $PQ$ model, that maps active and reactive powers to electrical phase angle and voltage magnitude. This representation typically appears in the power systems literature with simplified device dynamics \cite{devane_PrimaryFrequency_17, watson_ControlInterlinking_21, haberle_DecentralizedParametric_25, dey_PassivityBasedDecentralized_23, cifelli_DecentralizedSmall_25}, but here we allow for full-order device models. We then show that the $PQ$ model is related to the impedance representation of Part I (here referred to as the $IV$ model) by a static feedback transformation. 

Unstable bus impedances --- or, equivalently, non-minimum-phase bus admittances --- have been identified previously \cite{wen_InverseNyquist_17, liao_ImpedanceBasedStability_20, chen_LimitationsUsing_24}, but their origin has not been systematically characterised. By applying the Nyquist theorem to the loop linking the $IV$ and $PQ$ models, we show that such instabilities arise generically: droop-controlled buses typically exhibit low-frequency unstable $IV$ poles when the equilibrium reactive power is positive, while $PQ$ subsystems may exhibit high-frequency unstable poles when the bus contains a series passive impedance, such as a transformer or an inverter filter.
\lhg{The presence of these unstable poles greatly complicates the application of decentralised stability techniques as they must be stabilised by the network interconnection, which introduces centralised aspects into the analysis.}

\lhg{Because the mechanisms causing subsystem instabilities in the $IV$ and $PQ$ models act on well-separated timescales, we} propose a stable \textit{hybrid} subsystem \lhi{representation} that leverages a dynamic loop transformation so that each bus possesses the properties of the $PQ$ model at low frequencies and the $IV$ model at high frequencies. 
Similar frequency-dependent loop transformations have been proposed previously to recover passivity \cite{dey_PassivityBasedDecentralized_23} or phase sectoriality \cite{cifelli_DecentralizedSmall_25} at low frequencies, but here the bus transformation is constructed in feedback to remove open-loop instabilities, while the network system is modified additively and so remains stable and sparse. 

\lhg{The quadratic constraint-based framework of Part I is then applied to the hybrid model with stable subsystems,}
\lhi{which allows the use of} the extended-subsystem conditions and their plug-and-play capabilities. \lhg{The transformation to the hybrid model alters the network, however, so multipliers written for the untransformed network must be re-expressed for the new system. This changes how the homotopy parameter enters, which disrupts the property that allowed Part I to check conditions at a single homotopy point. We show that the Part I conditions nonetheless apply with a modified multiplier, and that the bus-level constraints need only be checked at the two endpoints of the homotopy.}

We next analyse the \lhg{low-frequency} Nyquist plot of the hybrid \lhi{representation}. \lhg{To the best of our knowledge, there are no existing decentralised results that deal with this range for systems without stiff voltage sources. We}
show that the characteristic loci decouple into sets of bounded and unbounded branches corresponding to dominant $P$--$\delta$ and $Q$--$V$ dynamics. Sufficient conditions on each set ensure that no encirclements of the point $-1$ occur, and these combine with the decentralised conditions at higher frequencies to certify grid stability. Moreover, we argue both sets are well behaved under standard network configurations and setpoints. In particular, the bounded branches are governed by a matrix closely related to the reduced power-flow Jacobian that operators already constrain during dispatch \cite{kundur_PowerSystem_22}. Centralised low-frequency analysis can therefore often be avoided, allowing decentralised grid codes to focus on more problematic high-frequency ranges.

Finally, the results of Part II are applied to the Kundur two-area system \cite{kundur_PowerSystem_22}, complementing the inverter-based case study of Part I. We show that decentralised conditions offering plug-and-play functionality \lhh{hold} at high frequencies, while centralised analysis at lower frequencies confirms that the bounded and unbounded branches are well behaved, \lhh{together certifying stability}.

Part II is structured as follows. \Cref{sec:background} summarises the necessary background theory. \Cref{sec:model} recaps the $IV$ model, introduces the $PQ$ model, shows how unstable subsystems arise, and presents the hybrid representation. \Cref{sec:stab} applies the stability framework of Part I to the hybrid model and \cref{sec:low_freq} analyses the low-frequency Nyquist plot, before the results are combined in \cref{sec:comb}. The case study is presented in \cref{sec:case_study} before concluding in \cref{sec:conclusion}. 

\section{Mathematical Background} \label{sec:background}
\subsection{Notation and Definitions}
The sets $\mathbb{R}$, $\mathbb{C}$, $\mathbb{C}_+$,  $\bar{\mathbb{C}}_+$, $\mathbb{H}^n$, $\Re(\cdot)$, and $\Im(\cdot)$ denote the real numbers, complex numbers, open right half-plane, closed right half-plane, Hermitian matrices, real part, and imaginary part, respectively. $I_n$ and $0_n$ denote the identity and zero matrices and $\mathbf{1}_n$ denotes the $n$-dimensional vector of ones. Define the matrix $J = \begin{bsmallmatrix} 0 & -1 \\ 1 & 0\end{bsmallmatrix}$. For matrices $A$, $B$, we let $A \otimes B$ denote the Kronecker product and $\sigma(A)$ denote the spectrum. For $A \in \mathbb{H}^n$, $A > 0$ ($A \geq 0$) denotes positive (semi)definiteness. For $A \in \mathbb{C}^{n \times n}$, we also use the numerical range $\mathcal{W}(A):= \{x^\ast A x | x \in \mathbb{C}^n, \lVert x \rVert = 1 \}$ [Part I, \cref{pt1-def:numerical_range}].
For an ordered indexed set $\mathcal{K}=\{\kappa_1, \ldots, \kappa_{|\mathcal{K}|}\}$, define $\oplus_{\kappa_k \in \mathcal{K}} B_k = \mathrm{diag} (B_1, \ldots, B_{|\mathcal{K}|})$ and $[a_k]_{\kappa_k \in \mathcal{K}} =[a_1^T, \ldots,a_{|\mathcal{K}|}^T]^T$ for matrices $B_k \in \mathbb{C}^{m_k \times n_k}$ and vectors $a_k \in \mathbb{C}^{n_k}$ associated with each $\kappa_k \in \mathcal{K}$.
\begin{definition} \label{def:group_inv}
    Let $A \in \mathbb{C}^{n\times n}$ be a matrix of index%
    \footnote{Defined as the smallest non-negative integer $k$ such that $\mathrm{rank}(A^k) = \mathrm{rank}(A^{k+1})$.}
    $1$. Then the \textit{group inverse} \cite{campbell_GeneralizedInverses_09, rose_LaurentExpansion_78} is defined as the unique matrix $A^\# \in \mathbb{C}^{n\times n}$ satisfying
    \begin{equation*}
        A A^\# A = A, \qquad A^\# A A^\# = A^\#, \qquad A A^\# = A^\# A.
    \end{equation*}
\end{definition}

The group inverse is the inverse of $A$ restricted to $\mathrm{range}(A)$ \cite{campbell_GeneralizedInverses_09}.
Let $P_0 \in \mathbb{C}^{n \times n}$ denote 
the oblique projector onto $\mathrm{null}(A)$ along 
$\mathrm{range}(A)$, so that
\begin{equation} \label{eq:group_inv_proj}
    \begin{gathered}
        \mathrm{range}(P_0) = \mathrm{null}(A), \quad
        \mathrm{null}(P_0) = \mathrm{range}(A), \\
        P_0^2 = P_0, \quad A P_0 = P_0 A = 0.
    \end{gathered}
\end{equation}
It then follows that 
\begin{equation} \label{eq:group_invser_proj_prop}
    A^\# P_0 = P_0 A^\# = 0, \qquad A A^\# = A^\# A = I - P_0.
\end{equation}

\subsection{Nyquist Theory}
Here, we restate the necessary linear systems theory from Part I. Recall the sets $\mathbf{R}^{m\times n}$, $\mathbf{RH}_\infty^{m \times n}$, and $\mathbf{RH}_{\infty,0}^{m \times n}$ are the sets of $m \times n$ proper, \lhb{real-}rational transfer functions, the subset containing no poles in $\bar{\mathbb{C}}_+$, and the subset containing no poles in $\bar{\mathbb{C}}_+ \setminus \{0\}$, respectively. For $P(s)$ and $K(s)$, of dimension $n\times m$ and $m \times n$, $[P(s),K(s)]$ denotes their negative-feedback interconnection, which we assume is well-posed \cite{zhou_RobustOptimal_96}. We impose the following definition of stability. 
\begin{definition} \label{def:stability}
    We say the interconnection is \textit{stable} if $(I+P(s)K(s))^{-1} \in \mathbf{RH}_{\infty,0}^{n\times n}$ \lhh{and, at $s=0$, has} \lhg{at most one pole.}
\end{definition}

The modified Nyquist contour in \cref{fig:tikz_nyquist_mod} is given by
\begin{equation} \label{eq:nyquist_contour_mod}
    \Gamma_N^{\,\epsilon} = \lim_{R \rightarrow \infty}
    \left\{ \Gamma_{j\omega}^{\epsilon-} \cup \Gamma_\epsilon \cup \Gamma_{j\omega}^{\epsilon+} \cup \Gamma_R \right\},
\end{equation}
where
    $ \Gamma_{j\omega}^{\epsilon-} = \left\{ s = j\omega \ | \ \omega \in \left[-R,-\epsilon \right] \right\}$,
    $ \Gamma_\epsilon = \left\{s =  \epsilon e^{j \varphi} \ | \ \varphi \in \left[ -\frac{\pi}{2}, \frac{\pi}{2} \right] \right\}$,
    $ \Gamma_{j\omega}^{\epsilon+} = \left\{ s = j\omega \ | \ \omega \in \left[\epsilon,R \right] \right\}$, 
    $\Gamma_R = \left\{s = Re^{j\varphi} \ | \ \varphi \in [\frac{\pi}{2}, -\frac{\pi}{2}]\right\}$,
and $\epsilon>0$ is sufficiently small. We let $\Gamma_N$ represent the contour shown in \cref{fig:tikz_nyquist_std} where $\epsilon = 0$. The multivariate Nyquist stability theorem [Part I, \cref{pt1-thm:nyquist}], \cite{desoer_GeneralizedNyquist_80} is used to conclude stability of the feedback configuration $[P(s),K(s)]$ under \cref{def:stability}. 

\begin{figure}[t]
    \centering
    \begin{subfigure}{0.2\textwidth}
        \centering
        \begin{tikzpicture}[>=Latex, line cap=round,scale=0.4]
\small
\def\R{3.2}      
\def\eps{0}   

\tikzset{
  nyqpath/.style={
    very thick,
    postaction={decorate},
    decoration={markings,
      mark=at position 0.15 with {\arrow{>}},
      mark=at position 0.28 with {\arrow{>}},
      mark=at position 0.55 with {\arrow{>}},
      mark=at position 0.9  with {\arrow{>}}
    }
  }
}

\draw[->] (-1.0,0) -- (\R+0.9,0) node[below right] {$\Re$};
\draw[->] (0,-\R-0.9) -- (0,\R+0.9) node[above left] {$\Im$};

\node[above left] at (0,\R) {$j\infty$};
\node[below left] at (0,-\R) {$-j\infty$};

\draw[nyqpath]
  (0,-\R)
    -- (0,\eps)
    arc[start angle=-90, end angle=90, radius=\eps cm]
    -- (0,\R)
    arc[start angle=90, end angle=-90, radius=\R cm];

\node at (\R*1.2,0.5) {$\Gamma_R$};
\node at (-0.75,\R*0.75) {$\Gamma_{j\omega}$};


\end{tikzpicture}
        \caption{$\Gamma_N$.}
        \label{fig:tikz_nyquist_std}
    \end{subfigure}
    \begin{subfigure}{0.2\textwidth}
        \centering
        \begin{tikzpicture}[>=Latex, line cap=round,scale=0.4]
\small
\def\R{3.2}      
\def\eps{0.5}   

\tikzset{
  nyqpath/.style={
    very thick,
    postaction={decorate},
    decoration={markings,
      mark=at position 0.13 with {\arrow{>}},
      mark=at position 0.34 with {\arrow{>}},
      mark=at position 0.55 with {\arrow{>}},
      mark=at position 0.9  with {\arrow{>}}
    }
  }
}

\draw[->] (-1.0,0) -- (\R+0.9,0) node[below right] {$\Re$};
\draw[->] (0,-\R-0.9) -- (0,\R+0.9) node[above left] {$\Im$};

\node[above left] at (0,\R) {$j\infty$};
\node[below left] at (0,-\R) {$-j\infty$};

\draw[nyqpath]
  (0,-\R)
    -- (0,-\eps)
    arc[start angle=-90, end angle=90, radius=\eps cm]
    -- (0,\R)
    arc[start angle=90, end angle=-90, radius=\R cm];

\node at (\R*1.2,0.5) {$\Gamma_R$};
\node at (2*\eps,0.5) {$\Gamma_{\epsilon}$};
\node at (-0.75,\R*0.75) {$\Gamma_{j\omega}^+$};
\node at (-0.75,-\R*0.75) {$\Gamma_{j\omega}^-$};


\end{tikzpicture}
        \caption{$\Gamma_N^{\,\epsilon}$.}
        \label{fig:tikz_nyquist_mod}
    \end{subfigure}
    \caption{Illustrations of the Nyquist contour and modified Nyquist contour with an indentation at the origin.}
    \label{fig:tikz_nyquist}
\end{figure}
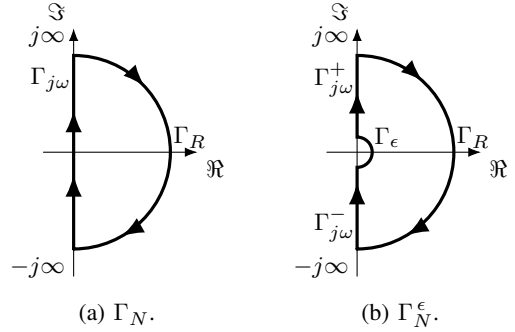

For \lhg{the open-loop transfer function (or \textit{return ratio})} $L(s) = P(s)K(s) \in \mathbf{RH}_{\infty,0}^{n \times n}$, a sufficient condition to demonstrate stability of the (well-posed) feedback system $[P(s),K(s)]$ under \cref{def:stability} (where, \lhh{at $s=0$}, $(I+L(s))^{-1}$ has at most \lhg{one} pole) is to show there exists an $\bar{\epsilon}>0$ such that for all $\epsilon \in (0,\bar{\epsilon}]$, [Part I, \cref{pt1-lemma:nyquist_sufficient}]
\begin{equation} \label{eq:nyquist_condition}
    -1 \notin \sigma(\tau L(s)), \quad \forall s \in \Gamma_N^{\,\epsilon}, \ \forall \tau \in [0,1].
\end{equation}
The condition \eqref{eq:nyquist_condition} holds at $s \in \Gamma_N^{\,\epsilon}$ if, for all $\tau \in [0,1]$, there exists a \textit{multiplier} $\Pi(s) \in \mathbb{C}^{(n+m) \times (n+m)}$ such that the following conditions are satisfied [Part I, \cref{pt1-lemma:iqc}]: 
\begin{subequations} \label{eq:iqc_conditions}
     \begin{align}
         \begin{bmatrix}
            P(s) \\ I_m
        \end{bmatrix}^\ast & 
        \begin{bmatrix}
            \Pi_{11}(s) & \Pi_{12}(s)^\ast \\ \Pi_{12}(s) & \Pi_{22}(s)
        \end{bmatrix}
        \begin{bmatrix}
            P(s) \\ I_m
        \end{bmatrix} > 0, \label{eq:iqc_conditions_top}%
        \\
         \begin{bmatrix}
            I_n \\ -\tau K(s)
        \end{bmatrix}^\ast &
        \begin{bmatrix}
            \Pi_{11}(s) & \Pi_{12}(s)^\ast \\ \Pi_{12}(s) & \Pi_{22}(s)
        \end{bmatrix}
        \begin{bmatrix}
            I_n \\ - \tau K(s)
        \end{bmatrix} \leq 0, \label{eq:iqc_conditions_btm}
    \end{align}
\end{subequations}
where $\Pi_{11}(s) \in \mathbb{H}^n$, $\Pi_{12}(s) \in \mathbb{C}^{m \times n}$ and $\Pi_{22}(s) \in \mathbb{H}^m$. 
The inequality \eqref{eq:iqc_conditions_btm} holds for all $\tau \in [0,1]$ if it holds at $\tau = 1$ and both $\Pi_{11}(s) \leq 0$ and $\Pi_{22}(s) \geq 0$ [Part I, \cref{pt1-cor:iqc_convex}]. 

\section{Grid Model} \label{sec:model}
\lhc{This section reviews the impedance representation of Part I, hereafter referred to as the $IV$ model, and introduces the power representation, termed the $PQ$ model, \lhg{where the subsystems are reformulated with respect to different inputs and outputs associated with real and reactive power}. The $PQ$ model is a natural alternative in power systems since droop and power-synchronisation laws often couple active and reactive \lhi{power} to frequency and voltage magnitude. Linking the two models through a static feedback transformation allows us to use Nyquist analysis in \cref{sec:model_instab} to show that each model may exhibit unstable poles \lhg{in the underlying subsystems} in many common scenarios, \lhg{even when the aggregate system is stable and well behaved}. \Cref{sec:model_hybrid} then combines the two models to produce a stable \textit{hybrid} model that preserves network sparsity, allowing the stability framework of Part I to be applied and plug-and-play conditions to be \lhg{formulated}.}

\subsection{Review of \textit{IV} Model} \label{sec:model_iv}
The three-phase AC grid is modelled as the directed graph $(\mathcal{V}_N,\mathcal{E}_N)$, where $\mathcal{V}_N = \mathcal{V}_B \cup\{\nu_0\}$, with $\mathcal{V}_B := \{\nu_1, \ldots, \nu_{N_B}\}$ the set of $N_B$ buses and $\nu_0$ the ground node, and $\mathcal{E}_N = \mathcal{E}_P \cup \mathcal{E}_0 $, with $\mathcal{E}_P$ the set of $N_P$ power lines (with elements $(\nu_i,\nu_j)$ denoting a line from $\nu_j$ to $\nu_i$) and $\mathcal{E}_0$ the set of loads and shunt connections (with elements $(\nu_i,\nu_0)$).
\begin{assumption} \label{assump:connected}
    The subgraph $(\mathcal{V}_B,\mathcal{E}_P) \subseteq (\mathcal{V}_N,\mathcal{E}_N)$ is connected.
\end{assumption}

The interconnection structure of the power system is captured by the incidence matrix $B_N\in\mathbb{R}^{N_B \times N_P}$ for the bus--line subgraph $(\mathcal{V}_B,\mathcal{E}_P)$, with $(i,k)^{\mathrm{th}}$ entry given by
\begin{equation} \label{eq:incidence_matrix}
    B_N^{(i,k)} = \begin{cases}
        1, & \text{if $\varepsilon_k \equiv (\nu_i,\nu_j) \in \mathcal{E}_P$, for some $\nu_j$}, \\
        -1, & \text{if $\varepsilon_k \equiv (\nu_j,\nu_i) \in \mathcal{E}_P$, for some $\nu_j$}, \\
        0, & \text{otherwise.}
    \end{cases}
\end{equation}

At each bus $\nu_i \in \mathcal{V}_B$, we associate a balanced three-phase voltage and current injection, which can be translated to a common $DQ$ reference frame rotating at constant angular frequency $\omega_0 >0$ to obtain $v_i^{DQ}(t) \in \mathbb{R}^2$ and $i_i^{DQ}(t) \in \mathbb{R}^2$. Similarly, with each bus $\nu_i \in \mathcal{V}_B$, we associate a \textit{local} $dq$ reference frame rotating at the dynamic local angular frequency $\omega_i(t)>0$. A three-phase signal $x_i^{dq}(t) \in \mathbb{R}^2$ in the local reference frame is translated to the common frame via $x_i^{DQ}(t) = T(\delta_i(t))x_i^{dq}(t)$, where
\begin{equation} \label{eq:rotation_matrix}
    T(\delta_i(t)) = \begin{bmatrix}
        \cos(\delta_i(t)) & -\sin(\delta_i(t)) \\ \sin(\delta_i(t)) & \cos(\delta_i(t))
    \end{bmatrix},
\end{equation}
and $\delta_i(t)$ is the instantaneous angular difference in reference frames given as the solution to 
\begin{equation} \label{eq:delta_dot}
    \dot{\delta}_i(t) = \omega_i(t) - \omega_0,
\end{equation}
which implies that $\omega_i(t) = \omega_0$ at equilibrium. Without loss of generality, we define the local $d$-axis to be aligned with the bus voltage vector so that $v_i^{dq}(t) = [V_i(t), 0]^T$, where $V_i(t) > 0$ is the voltage magnitude. Therefore, using \eqref{eq:rotation_matrix}, we obtain
\begin{equation} \label{eq:vDQ_in_polar}
    \begin{bmatrix}
        v_i^D(t) \\ v_i^Q(t)
    \end{bmatrix} = \begin{bmatrix}
        V_i(t) \cos(\delta_i(t)) \\ V_i(t) \sin(\delta_i(t))
    \end{bmatrix}.
\end{equation}

We define the reference-frame invariant net active- and reactive-power flows at bus $\nu_i \in \mathcal{V}_B$ as
\begin{subequations} \label{eq:powers}
    \begin{align}
        P_i &= {v_i^{DQ}}^T i_i^{DQ} = v_i^Di_i^D + v_i^Qi_i^Q, \label{eq:active_power} \\
        Q_i & = {v_i^{DQ}}^T J i_i^{DQ} = v_i^Qi_i^D - v_i^Di_i^Q. \label{eq:reactive_power}
    \end{align}
\end{subequations}

As described in Part I, at each $\nu_i \in \mathcal{V}_B$, \lhc{we linearise the dynamics of the connected device about the network equilibrium in the common $DQ$ reference frame}
to obtain the $2 \times 2$ \lhc{Laplace-domain} bus impedance $Z_i(s)$, defined by 
\begin{equation} \label{eq:bus_iv_relationship}
    \Delta v_i^{DQ}(s) = - Z_i(s) \Delta i_i^{DQ}(s).
\end{equation}

Similarly, each line, load, or shunt $(\nu_i,\nu_j) \in \mathcal{E}_N$ is modelled in the Laplace domain by the branch admittance, defined by
\begin{equation} \label{eq:line_dynamics_s}
    \Delta i_{ij}^{DQ}(s) = Y_{ij}(s) \left( \Delta v_i^{DQ}(s) - \Delta v_j^{DQ}(s) \right),
\end{equation}
where $i_{ij}$ is the associated branch current. 
For a load or shunt in $\mathcal{E}_0$, $\Delta v_0^{DQ}(s) = [0,0]^T$, so we define $Y_i(s):= Y_{i0}(s)$. If a load or shunt at $\nu_i \in \mathcal{V}_B$ is absent, we let $Y_i(s) = 0_2$. 

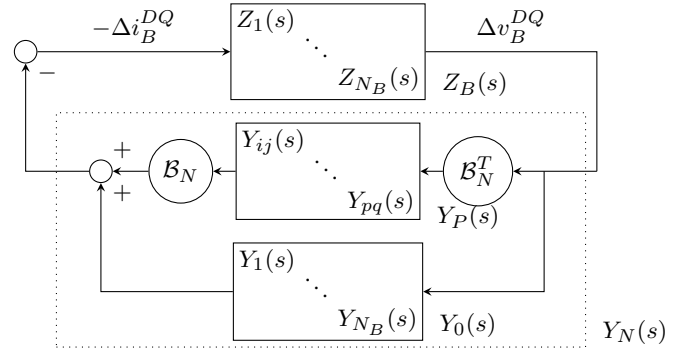
\begin{figure}[t]
\centering
\begin{tikzpicture}[scale=1]
\small

\node[matrix, draw,
    row sep=1pt, column sep=1pt,
    inner sep=1pt]
    (sys_Z) at (0,0)
    {
        \node{$Z_1(s)$}; & & \\[-8pt]
        & \node{$\ddots$}; & \\
        & & \node{$Z_{N_B}(s)$}; \\
    };
    \node at (sys_Z.south east)[right=3pt,yshift=6pt]{$Z_B(s)$};
\node[matrix, draw=black,
    below=of sys_Z,yshift=+20pt,
    row sep=1pt, column sep=1pt,
    inner sep=1pt]
    (sys_Yp) 
    {
        \node{$Y_{ij}(s)$}; & & \\[-8pt]
        & \node{$\ddots$}; & \\
        & & \node{$Y_{pq}(s)$}; \\
    };
    \node at (sys_Yp.south east)[right=3pt,yshift=1pt]{$Y_P(s)$};
\node[matrix, draw=black,
    below=of sys_Yp,yshift=+20pt,row sep=1pt, column sep=1pt,
    inner sep=1pt]
    (sys_Y0) 
    {
        \node{$Y_1(s)$}; & & \\[-8pt]
        & \node{$\ddots$}; & \\
        & & \node{$Y_{N_B}(s)$}; \\
    };
    \node at (sys_Y0.south east)[right=3pt,yshift=6pt]{$Y_0(s)$};

\node[draw, circle, minimum size=0.75cm,
            left=of sys_Yp,xshift=20pt]
            (B)  
            {$\mathcal{B}_N$};
\node[draw, circle, minimum size=0.75cm,
            right=of sys_Yp,xshift=-20pt]
            (Bt)   {$\mathcal{B}_N^T$};

\node[draw, circle, minimum size=0.05cm,
    left=of B,xshift=15pt]
    (left1) 
    {};
    \node[yshift=5pt,xshift=5pt] at (left1.north east){$+$};
    \node[yshift=-5pt,xshift=5pt] at (left1.south east){$+$};
\node[right= of Bt,xshift=-20pt] (right1){};
\node[right= of right1,xshift=-15pt] (right2){};
\node[draw, circle, minimum size=0.05cm]
    (neg) at (-4,0)
    {};
    \node[yshift=-5pt,xshift=5pt] at (neg.south east){$-$};

\node[draw,black,thin,dotted, anchor=north west,
    minimum width =7cm, minimum height =3.1cm]
    (box) at (-3.6,-0.8) {};
    \node at (box.south east)[right=3pt,yshift=6pt]{$Y_N(s)$};

\draw[-stealth] (neg.east) -- (sys_Z.west)
    node[midway,above]{$-\Delta i_B^{DQ}$};
\draw[] (sys_Z.east) -| (right2.center)
    node[pos=0.25,above]{$\Delta v_B^{DQ}$};
\draw[-stealth] (right2.center) -- (Bt.east);
\draw[-stealth] (right1.center) |- (sys_Y0.east);
\draw[-stealth] (Bt.west) -- (sys_Yp.east);
\draw[-stealth] (sys_Yp.west) -- (B.east);
\draw[-stealth] (B.west) -- (left1.east);
\draw[-stealth] (sys_Y0.west) -| (left1.south);
\draw[-stealth] (left1.west) -| (neg.south);

\end{tikzpicture}
\caption{Negative-feedback interconnection of the bus impedances in $Z_B(s)$ and the network admittance $Y_N(s)$.}
\label{fig:tikz_model_iv}
\end{figure}

Letting $\Delta v_B^{DQ}(s) = [\Delta v_i^{DQ}(s)]_{\nu_i \in \mathcal{V}_B}$ and $\Delta i_B^{DQ}(s) = [\Delta i_i^{DQ}(s)]_{\nu_i \in \mathcal{V}_B}$, Kirchhoff's current law allows us to write
\begin{equation} \label{eq:kcl_matrix_form}
    \Delta i_B^{DQ}(s) = Y_N(s) \Delta v_B^{DQ}(s),
\end{equation}
where $Y_N(s)$ is the stable $2N_B \times 2N_B$ sparse network admittance matrix given by
\begin{equation} \label{eq:network_admittance}
    Y_N(s) = \mathcal{B}_N Y_P(s) \mathcal{B}_N^T + Y_0(s),
\end{equation}
where $\mathcal{B}_N = B_N \otimes I_2$ (with $B_N$ defined in \eqref{eq:incidence_matrix}), $Y_P(s) = \oplus_{(\nu_i,\nu_j) \in \mathcal{E}_P} Y_{ij}(s)$, and $Y_0(s) = \oplus_{\nu_i \in \mathcal{V}_B} Y_i(s)$. Similarly, letting $Z_B(s) = \oplus_{\nu_i \in \mathcal{V}_B} Z_i(s) $, we obtain
\begin{equation} \label{eq:kvl_matrix_form}
    \Delta v_B^{DQ} = - Z_B(s) \Delta i_B^{DQ}.
\end{equation}
Together, \eqref{eq:kcl_matrix_form} and \eqref{eq:kvl_matrix_form} form the negative-feedback interconnection $[Z_B(s),Y_N(s)]$ between the bus dynamics and the network admittance matrix, as illustrated in \cref{fig:tikz_model_iv}.

\subsection{PQ Model} \label{sec:model_pq}
As an alternative to the $IV$ model, we now model the grid as a feedback interconnection coupling the bus active and reactive powers to electrical phase angle and voltage magnitude.

First, at each $\nu_i \in \mathcal{V}_B$, we linearise \eqref{eq:powers} with respect to $v_i^{DQ}$ and $i_i^{DQ}$ and convert to the Laplace domain to get
\begin{equation} \label{eq:pq_deriv_1}
    \begin{aligned}
        \begin{bmatrix}
            \Delta P_i(s) \\ \Delta Q_i(s)
        \end{bmatrix}
        & = \begin{bmatrix}
            i_i^{D\star} &  i_i^{Q\star}  \\ -i_i^{Q\star} &  i_i^{D\star}
        \end{bmatrix} \Delta v_i^{DQ}(s) \! \\
        &\qquad+ \! \begin{bmatrix}
            v_i^{D\star} &  v_i^{Q\star} \\ v_i^{Q\star} & - v_i^{D\star}
        \end{bmatrix} \Delta i_i^{DQ}(s),
    \end{aligned}
\end{equation}
where $(\cdot)^\star$ represents the equilibrium value of the respective variable. Next, linearising \eqref{eq:vDQ_in_polar} gives
\begin{equation} \label{eq:pq_deriv_2}
    \begin{aligned}
        \Delta v_i^{DQ}(s)
        & = \begin{bmatrix}
            -V_i^\star \sin(\delta_i^\star) & \cos(\delta_i^\star) \\
            V_i^\star \cos(\delta_i^\star) & \sin(\delta_i^\star)
        \end{bmatrix} \begin{bmatrix}
            \Delta \delta_i(s) \\ \Delta V_i(s)
        \end{bmatrix} \\
        & = \begin{bmatrix}
            -v_i^{Q\star} & v_i^{D\star} \\ v_i^{D\star} & v_i^{Q\star}
        \end{bmatrix} \begin{bmatrix}
            \Delta \delta_i(s) \\ \frac{\Delta V_i(s)}{V_i^\star}
        \end{bmatrix}.
    \end{aligned}
\end{equation}
Using \eqref{eq:pq_deriv_2}, the first term in \eqref{eq:pq_deriv_1} becomes
\begin{equation} \label{eq:pq_deriv_3}
    \begin{aligned}
        \begin{bmatrix}
                i_i^{D\star} & i_i^{Q\star} \\ -i_i^{Q\star} &  i_i^{D\star}
            \end{bmatrix} \Delta v_i^{DQ}(s)
            = \begin{bmatrix}
                -Q_i^\star & P_i^\star \\ P_i^\star & Q_i^\star
            \end{bmatrix} \begin{bmatrix}
            \Delta \delta_i(s) \\ \frac{\Delta V_i(s)}{V_i^\star}
        \end{bmatrix}.
    \end{aligned}
\end{equation}
Now, define the matrices appearing in \eqref{eq:pq_deriv_1}–\eqref{eq:pq_deriv_3} as
\begin{equation} \label{eq:UW_defs}
    \begin{gathered}
        \begin{aligned}
            U_i^\dagger := \begin{bmatrix} v_i^{D\star} & v_i^{Q\star} \\ v_i^{Q\star} & -v_i^{D\star} \end{bmatrix}, &&
            U_i^\ddagger := \begin{bmatrix} -v_i^{Q\star} & v_i^{D\star} \\ v_i^{D\star} & v_i^{Q\star} \end{bmatrix},
        \end{aligned} \\
        W_i := \begin{bmatrix} -Q_i^\star & P_i^\star \\ P_i^\star & Q_i^\star \end{bmatrix},
    \end{gathered}
\end{equation}
where we note $U_i^\dagger$ and $U_i^\ddagger$ are symmetric and
\begin{equation} \label{eq:U_squared}
    U_i^\dagger U_i^\dagger = U_i^\ddagger U_i^\ddagger = {V_i^\star}^2I_2,
\end{equation}
so $(U_i^\dagger)^{-1} = \frac{1}{{V_i^\star}^2} U_i^\dagger$ and $(U_i^\ddagger)^{-1} = \frac{1}{{V_i^\star}^2} U_i^\ddagger$.
The matrix $W_i$ is zero if and only if the bus carries no power. Furthermore, $\mathrm{tr}(W_i) = 0$ and $\det(W_i) = -{S_i^\star}^2$, where $S_i^\star := \sqrt{{P_i^\star}^2 + {Q_i^\star}^2}$ is the apparent power at $\nu_i \in \mathcal{V}_B$, so the eigenvalues of $W_i$ are $\pm S_i^\star$.
Let
\begin{equation*}
    \Delta S_i(s) := [\Delta P_i(s), \Delta Q_i(s)]^T\!,
    \ \:
    \Delta \phi_i(s) := [\Delta \delta_i(s), \tfrac{\Delta V_i(s)}{{V_i^\star}}]^T\!,
\end{equation*}
and define $U^\dagger = \oplus_{\nu_i \in \mathcal{V}_B}U_i^\dagger$, $U^\ddagger = \oplus_{\nu_i \in \mathcal{V}_B}U_i^\ddagger$, $W = \oplus_{\nu_i \in \mathcal{V}_B}W_i$, $\Delta S(s) = [\Delta S_i(s)]_{\nu_i \in \mathcal{V}_B}$, and $\Delta \phi(s) = [\Delta \phi_i(s)]_{\nu_i \in \mathcal{V}_B}$. 
Using \eqref{eq:network_admittance}, the network dynamics defining the active and reactive power injections at each bus are therefore given by
\begin{equation*}
    \begin{split}
        \Delta S(s)
        & = U^\dagger Y_N(s) \Delta v_B^{DQ}(s) +  W \Delta \phi(s) \\
        & = (U^\dagger Y_N(s) U^\ddagger + W) \Delta \phi(s).
    \end{split}
\end{equation*}
The network dynamics are therefore given by the transfer function
\begin{equation} \label{eq:N_def}
    N_{PQ}(s) := U^\dagger Y_N(s) U^\ddagger + W,
\end{equation}
which we note \lhg{inherits stability and sparsity from $Y_N(s)$, as $U^\dagger$, $U^\ddagger$ and $W$ are block-diagonal and static. At low frequencies, the entries of $N_{PQ}(s)$ are related to the power-flow equations, as shown in \cref{sec:low_freq}.}

\begin{figure}[t]
\centering
\begin{tikzpicture}
\small


\node[matrix, draw,
    row sep=1pt, column sep=1pt,
    inner sep=1pt]
    (sys_Z) at (0,0)
    {
        \node{$G_1(s)$}; & & \\[-8pt]
        & \node{$\ddots$}; & \\
        & & \node{$G_{N_B}(s)$}; \\
    };
    \node at (sys_Z.south east)[right=3pt,yshift=6pt]{$G_B(s)$};

\node[draw=black,
    below=of sys_Z,yshift=0.8cm,
    minimum size = 1cm]
    (sys_Yp) 
    {$Y_N(s)$};
    
\node[draw=black,
    below=of sys_Yp,yshift=0.5cm]
    (sys_Y0) 
    {$ \bigoplus\limits_{\nu_i \in \mathcal{V}_B} \begin{bsmallmatrix}
                -Q_i^\star & P_i^\star \\ P_i^\star & Q_i^\star
            \end{bsmallmatrix}$};
    \node at (sys_Y0.south east)[right=3pt,yshift=6pt]{$W$};

\node[draw=black,
    left=of sys_Yp,xshift=0.8cm]
    (B) 
    {$\bigoplus\limits_{\nu_i \in \mathcal{V}_B}\begin{bsmallmatrix}
                v_i^{D\star} & v_i^{Q\star} \\ v_i^{Q\star} & - v_i^{D\star}
            \end{bsmallmatrix}$};
    \node at (B.south east)[left=1pt,yshift=-6pt]{$U^\dagger$};
\node[draw=black,
    right=of sys_Yp,xshift=-0.8cm]
    (Bt) 
    {$\bigoplus\limits_{\nu_i \in \mathcal{V}_B} \begin{bsmallmatrix}
                -v_i^{Q\star} & v_i^{D\star} \\ v_i^{D\star} & v_i^{Q\star}
            \end{bsmallmatrix}$};
    \node at (Bt.south east)[left=1pt,yshift=-6pt]{$U^\ddagger$};

\node[draw, circle, minimum size=0.05cm,
    left=of B,xshift=0.5cm]
    (left1) 
    {};
    \node[yshift=5pt,xshift=5pt] at (left1.north east){$+$};
    \node[yshift=-5pt,xshift=5pt] at (left1.south east){$+$};
\node[right= of Bt,xshift=-0.8cm] (right1){};
\node[right= of right1,xshift=-1cm] (right2){};
\node[draw, circle, minimum size=0.05cm]
    (neg) at (-4.4,0)
    {};
    \node[yshift=-5pt,xshift=5pt] at (neg.south east){$-$};

\node[draw,black,thin,dotted,anchor=north west,
    minimum width =8.2cm, minimum height =2.7cm]
    (box) at (-4.3,-0.7) {};
    \node at (box.south east)[left=3pt,yshift=7pt]{$N_{PQ}(s)$};

\draw[-stealth] (neg.east) -- (sys_Z.west)
    node[pos=0.5,above]{$-\begin{bmatrix} \Delta P_i(s) \\ \Delta Q_i(s)\end{bmatrix}_{\nu_i \in \mathcal{V}_B}$};
\draw[] (sys_Z.east) -| (right2.center)
    node[pos=0.25,above]{$\begin{bmatrix} \Delta \delta_i(s) \\ \nicefrac{\Delta V_i(s)}{V_i^\star} \end{bmatrix}_{\nu_i \in \mathcal{V}_B}$};
\draw[-stealth] (right2.center) -- (Bt.east);
\draw[-stealth] (right1.center) |- (sys_Y0.east);
\draw[-stealth] (Bt.west) -- (sys_Yp.east);
\draw[-stealth] (sys_Yp.west) -- (B.east);
\draw[-stealth] (B.west) -- (left1.east);
\draw[-stealth] (sys_Y0.west) -| (left1.south);
\draw[-stealth] (left1.west) -| (neg.south);

\end{tikzpicture}
\caption{Negative-feedback interconnection of the bus systems in $G_B(s)$ and the network system $N_{PQ}(s)$.}
\label{fig:tikz_model_pq}
\end{figure}
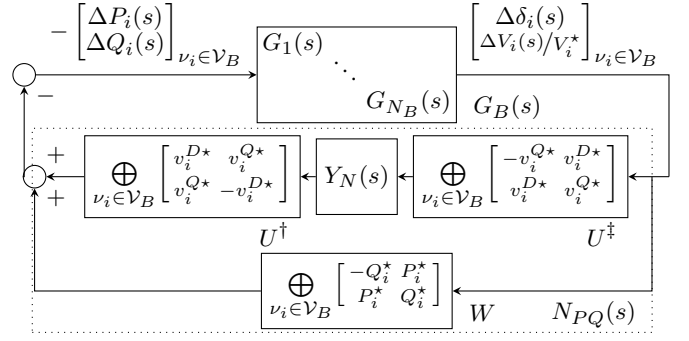

At each bus $\nu_i \in \mathcal{V}_B$, we define the transfer function $G_i(s)$ \lhg{that contains the aggregate bus dynamics}:
\begin{equation} \label{eq:bus_pq_relationship}
    \begin{bmatrix}
        \Delta \delta_i(s) \\ \frac{\Delta V_i(s)}{V_i^\star}
    \end{bmatrix} = -G_i(s) \begin{bmatrix}
        \Delta P_i(s) \\ \Delta Q_i(s)
    \end{bmatrix},
\end{equation}
and let $G_B(s) = \oplus_{\nu_i \in \mathcal{V}_B} G_i(s)$. From \eqref{eq:delta_dot}, we have $\Delta \delta_i(s) =
\frac{1}{s} \Delta \omega_i(s)$, so $G_i(s)$ takes the form
\begin{equation} \label{eq:pq_transfer_function}
    G_i(s) = \begin{bmatrix}
        \frac{1}{s} G_{\omega p, i}(s) & \frac{1}{s} G_{\omega q,i}(s) \\ G_{vp,i}(s) & G_{vq,i}(s)
    \end{bmatrix},
\end{equation}
\lhc{where the first subscript indicates the output (frequency or normalised voltage magnitude) and the second subscript indicates the input (active or reactive power). Note that $G_i(s)$ has \lhg{one} pole at the origin.}

The bus and network systems are then arranged to obtain the negative-feedback interconnection $[G_B(s),N_{PQ}(s)]$, as illustrated in \cref{fig:tikz_model_pq}.

Furthermore, we can link the $IV$ and $PQ$ transfer functions at each $\nu_i \in \mathcal{V}_B$ by substituting \eqref{eq:pq_deriv_1}-\eqref{eq:pq_deriv_3} into \eqref{eq:bus_pq_relationship}
\lhc{and comparing with \eqref{eq:bus_iv_relationship} to obtain}
\begin{equation} \label{eq:iv_from_pq}
    Z_i(s) = U_i^\ddagger (I+G_i(s)W_i)^{-1} G_i(s) U_i^\dagger.
\end{equation}
This shows that the $IV$ bus impedance is obtained by placing $G_i(s)$ in negative feedback with $W_i$ and pre- and post-multiplying by $U_i^\ddagger$ and $U_i^\dagger$, as illustrated by \cref{fig:tikz_pq_to_iv_fb}. Similarly, using \eqref{eq:U_squared}, we invert \eqref{eq:iv_from_pq} to obtain
\begin{equation} \label{eq:pq_from_iv}
    G_i(s) =  (I-\check{Z}_i(s) W_i)^{-1} \check{Z}_i(s),
\end{equation}
where $\check{Z}_i(s) := \frac{1}{{V_i^\star}^4}U_i^\ddagger Z_i(s)U_i^\dagger$, giving $G_i(s)$ as the positive-feedback interconnection of $\check{Z}_i(s)$ and $W_i$. We can therefore conclude that a conversion between $PQ$ and $IV$ models can be performed by using a loop transformation involving the matrices $U_i^\dagger$, $U_i^\ddagger$, and $W_i$.

\begin{figure}[t]
\centering
\begin{tikzpicture}
\small

\node[draw,minimum width =1.25cm, minimum height =0.6cm]
    (sys_Z) at (0,0)
    {$G_i(s)$};

\node[draw=black,
    above=of sys_Z,yshift=-0.8cm,
    minimum width =1.25cm, minimum height =0.6cm]
    (sys_Y0) 
    {$W_i$};

\node[draw, circle, minimum size=0.05cm,
    left=of sys_Z,xshift=0.5cm]
    (left1) 
    {};
    \node[yshift=5pt,xshift=5pt] at (left1.north east){$-$};
    \node[yshift=-5pt,xshift=-5pt] at (left1.south west){$+$};
\node[draw=black,
    left=of left1,xshift=0.5cm,
    minimum width =1.25cm, minimum height =0.6cm]
    (B) 
    {$U_i^\dagger$};
\node[right= of sys_Z,xshift=-0.5cm] (right1){};
\node[draw=black,
    right=of right1,xshift=-0.5cm,
    minimum width =1.25cm, minimum height =0.6cm]
    (Bt) 
    {$U_i^\ddagger$};
\node[right= of Bt,xshift=-0.5cm] (right2){};
\node[left=of B,xshift=0.5cm]
    (neg)
    {};

\node[draw,black,thin,dotted,anchor=north west,
    minimum width =7cm, minimum height =1.7 cm]
    (box) at (-3.6,1.2) {};
    \node at (box.north east)[left=3pt,yshift=-8pt]{$Z_i(s)$};

\draw[-stealth] (neg.east) -- (B.west);
\draw[-stealth] (B.east) -- (left1.west);
\draw[-stealth] (left1.east) -- (sys_Z.west);
    
\draw[] (sys_Z.east) -- (right1.center);
\draw[-stealth] (right1.center) -- (Bt.west);
\draw[-stealth] (Bt.east) -- (right2.center);
\draw[-stealth] (right1.center) |- (sys_Y0.east);
\draw[-stealth] (sys_Y0.west) -| (left1.north);


\end{tikzpicture}
\caption{Relationship between the $PQ$ and $IV$ bus transfer functions.}
\label{fig:tikz_pq_to_iv_fb}
\end{figure}
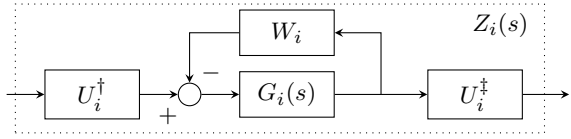

\subsection{Unstable Poles in \textit{IV} and \textit{PQ} Models} \label{sec:model_instab}
We now \lhc{use the feedback relationship between the $IV$ and $PQ$ models to show} that both representations can exhibit unstable \lhc{open-loop}\footnote{
\lhg{Here, the open-loop system refers to the dynamics where the feedback loop with the network interconnection in \cref{fig:tikz_model_iv} or \cref{fig:tikz_model_pq} is open. The individual bus systems may still contain local feedback control policies.}
} poles in common bus configurations. 
\lhc{This complicates decentralised analysis because the sufficient condition \eqref{eq:nyquist_condition} requires a \lhg{return ratio} with no poles in \lhh{$\bar{\mathbb{C}}_+ \setminus \{0\}$}. When subsystems are open-loop unstable, stabilisation is typically supplied by the interconnection itself, reintroducing centralised elements into the analysis.}
\lhc{It also explains the unstable impedances found in the case study of Part I, which forced the use of the admittance representation and blocked the extended-subsystem conditions that enable plug-and-play operation. We treat each representation in turn.}

\subsubsection{Low-frequency poles in \textit{IV} impedances} \label{sec:model_instab_iv}
Consider a bus $\nu_i \in \mathcal{V}_B$ that links frequency to active power and voltage magnitude to reactive power using a droop control law. The entries of the transfer function \eqref{eq:pq_transfer_function} for such a system are
\begin{equation} \label{eq:droop_tf}
    \begin{aligned}
        & G_{\omega p,i}(s) = \tfrac{k_{\omega,i}}{\tau_{\omega,i}s + 1}, &&
        G_{\omega q,i }(s) = 0, \\
        & G_{v p,i}(s) = 0, &&
        G_{v q,i}(s) = \tfrac{k_{v,i}}{\tau_{v,i}s + 1},
    \end{aligned}
\end{equation}
where $k_{\omega , i}, k_{v,i} >0$ are gains and $\tau_{\omega ,i}, \tau_{v,i} \geq 0$ are time constants.

As $G_i(s) \in \mathbf{RH}_{\infty,0}^{2 \times 2}$ \lhc{with \lhg{one} pole at the origin}, using [Part I, \cref{pt1-thm:nyquist}], we can determine the number of right half-plane poles in $Z_i(s)$ by counting the number of encirclements of the point $-1$ by the characteristic loci of $L_{Z,i}(s) = G_i(s)W_i$ as $s$ traverses the indented Nyquist contour $\Gamma_N^{\,\epsilon}$ in \eqref{eq:nyquist_contour_mod}.
In the following lemma, we evaluate the characteristic loci of $L_{Z,i}(s)$ as $s \rightarrow 0$.

\begin{lemma} \label{lemma:iv_unstable}
Consider a bus $\nu_i \in \mathcal{V}_B$ with $PQ$ transfer function $G_i(s)$ as in \eqref{eq:pq_transfer_function} with entries given by \eqref{eq:droop_tf}, \lhc{and with $Q_i^\star \neq 0$}. As $s \rightarrow 0$, the eigenvalues of the \lhg{return ratio} $L_{Z,i}(s) = G_i(s)W_i$ with $W_i$ in \eqref{eq:UW_defs} are given by
\begin{equation} \label{eq:instab_iv_loci}
    \begin{aligned}
        \lambda_-(s) &= -\tfrac{Q_i^\star k_{\omega,i}}{s} + Q_i^\star k_{\omega,i}\tau_{\omega,i} - \tfrac{{P_i^\star}^2 k_{v,i}}{Q_i^\star} + \mathcal{O}(s), \\
        \lambda_+(s) &= Q_i^\star k_{v,i} + \tfrac{{P_i^\star}^2 k_{v,i}}{Q_i^\star} + \mathcal{O}(s).
    \end{aligned}
\end{equation}
\end{lemma}
\begin{proof}
    The lemma is proved in Appendix \ref{app:iv_unstable_proof}.
\end{proof}

We therefore see that the characteristic loci split into a bounded and unbounded branch as $s \rightarrow 0$. In particular, for $s = j\omega$, the unbounded branch approaches infinity along the imaginary axis (ignoring $\mathcal{O}(1)$ corrections), with direction determined by the sign of the equilibrium bus reactive power $Q_i^\star$ calculated via \eqref{eq:reactive_power}.

Along $\Gamma_\epsilon$,
$
    s = \epsilon e^{j\varphi},  -\frac{\pi}{2} \leq \varphi \leq \frac{\pi}{2},
$
so we obtain
\begin{equation}
    \lambda_-(\epsilon e^{j\varphi}) = -Q_i^\star k_{\omega,i} \epsilon^{-1} e^{-j\varphi} + \mathcal{O}(1).
\end{equation}
If $Q_i^\star > 0$, then $\lambda_-(s)$ \lhc{makes a semicircular clockwise half-turn through the left half-plane, crossing the negative real axis at $-\frac{Q_i^\star k_{\omega,i}}{\epsilon} + \mathcal{O}(1)$. For all $\epsilon$ sufficiently small, this crossing lies to the left of $-1$, so the arc contributes a clockwise encirclement. Unless it is cancelled by a counter-clockwise encirclement at higher frequencies, [Part I, \cref{pt1-thm:nyquist}] then implies that $Z_i(s)$ in \eqref{eq:iv_from_pq} has a pole in $\mathbb{C}_+$.} 

As an example, consider the case $k_{\omega,i} = 2$, $\tau_{\omega,i} = 1$, $k_{v,i} = 0.5$, $\tau_{v,i} = 2$, $P_i^\star = \SI{1}{\pu}$, $Q_i^\star = \SI{0.5}{\pu}$, and ${v_i^{DQ}}^\star = [1,0]^T$ p.u. 
\lhc{\Cref{fig:plot_iv_unstable} shows that the unbounded branch encircles $-1$, so $Z_i(s)$ is unstable.}
Through direct calculation, we find that $Z_i(s)$ has a right half-plane pole at $s = 0.7684$.

\begin{figure}[t]
    \centering
    \includegraphics[width=0.48\textwidth]{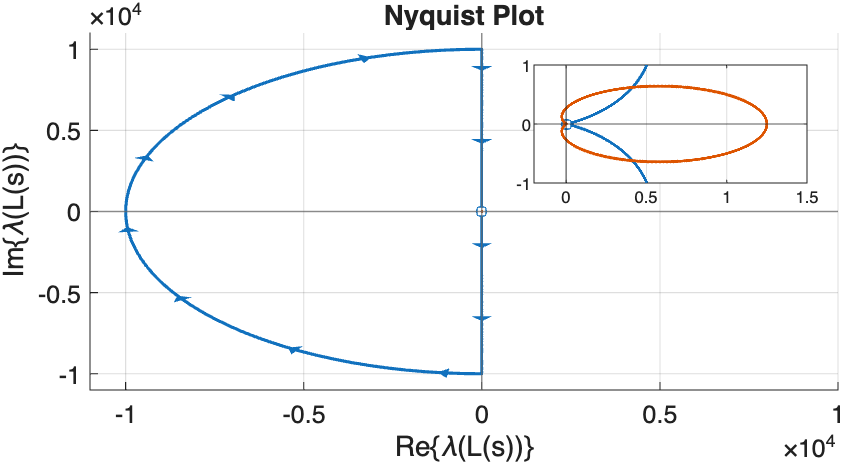}
    \caption{The Nyquist plot of $L_{Z,i}(s) = G_i(s)W_i$ along $\Gamma_N^{\,\epsilon}$ (with indentation of radius $\epsilon = 1\times 10^{-4}$) with entries of $G_i(s)$ given in \eqref{eq:droop_tf}. The inset rescales the plot to show the bounded branch.}
    \label{fig:plot_iv_unstable}
\end{figure}

\subsubsection{High-frequency poles in \textit{PQ} systems}
Now consider a voltage source in series with a passive impedance composed of $RLC$ elements, e.g., a transformer or the low-pass filter of an inverter, as shown in \cref{fig:tikz_series_impedance}. If the voltage source has impedance $Z_i^s(s)$ and the passive impedance has transfer function $Z_i^p(s)$, then the total bus impedance is given by
\begin{equation}
    Z_i(s) = Z_i^s(s) + Z_i^p(s).
\end{equation}
\lhc{Near the resonant frequency of the passive element, $Z_i^p(j\omega)$ typically dominates so that}
\begin{equation} \label{eq:Z_bus_approx_Z_passive}
    Z_i(j\omega) \approx Z_i^p(j\omega).
\end{equation}
Such a passive impedance takes the general form
\begin{equation} \label{eq:gen_passive_compent}
    Z_i^p(s)=\begin{bmatrix}
        a(s) & b(s) \\ -b(s) & a(s)
    \end{bmatrix}.
\end{equation}
As the system is passive, we know that $Z_i^p(s)$ has no unstable poles and is positive-real [Part I, \cref{pt1-def:positive_real}], i.e., $Z_i^p(j\omega) + Z_i^p(j\omega)^\ast \geq 0$, so that the eigenvalues of $Z_i^p(j\omega)$ have non-negative real part. We now show that when converted to a $PQ$ model using \eqref{eq:pq_from_iv}, the passive component \eqref{eq:gen_passive_compent} may have unstable poles.

\begin{figure}[t]
\centering
\begin{circuitikz}
\small
    \draw
        (0,1.5) 
        to[short] ++(1.5,0)
        to[cute inductor] ++(1,0)
        to[short] ++(0.25,0) 
        to[american resistor] ++(1,0) 
        to[short,-*] ++(2,0)
        to[open,-*] ++(0,-1.5)
        to[short] (0,0)
        to[sV,name=vs] (0,1.5)
        {};         
    \draw
        (4.25,1.5) to[capacitor] ++(0,-1.5) node[ground] 
        {};

    \node[draw, rectangle, dotted, anchor=center,
        label=left:$Z_i^s(s)$, 
        minimum width=1.2cm, 	minimum height=1.2cm]
        (Zs) at (vs.center) {};

    \node[draw, rectangle, dotted, anchor=north west,
        label=above:$Z_i^p(s)$, 
        minimum width=4cm, 	minimum height=2.0cm]
        (Zp) at (1,1.8) {};
   
\end{circuitikz}
\caption{A series connection of a voltage source $Z_i^s(s)$ with a passive impedance $Z_i^p(s)$.}
\label{fig:tikz_series_impedance}
\end{figure}
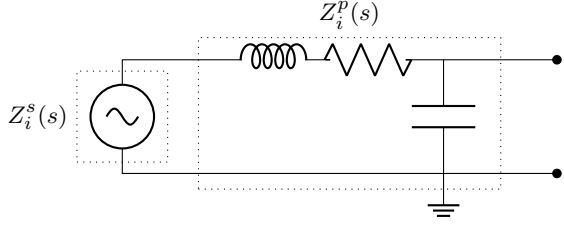

\begin{lemma} \label{lemma:pq_unstable}
Consider a passive impedance at $\nu_i \in \mathcal{V}_B$ with $IV$ transfer function $Z_i^p(s)$ as in \eqref{eq:gen_passive_compent}. The eigenvalues of the \lhg{return ratio} $L_{G,i}(s) = -\check{Z}_i^p(s)W_i$, with $\check{Z}_i^p(s) = \frac{1}{{V_i^\star}^4}U_i^\ddagger Z_i^p(s)U_i^\dagger$, where $U_i^\dagger$, $U_i^\ddagger$, and $W_i$ are defined in \eqref{eq:UW_defs}, are given by
\begin{equation} \label{eq:instab_pq_loci}
    \lambda_\pm(s) = \pm \frac{S_i^\star}{ {V_i^\star}^2 } \sqrt{a(s)^2 + b(s)^2}.
\end{equation}
\end{lemma}
\begin{proof}
    The lemma is proved in Appendix \ref{app:pq_unstable_proof}.
\end{proof}

\begin{figure}[t]
    \centering
    \includegraphics[width=0.48\textwidth]{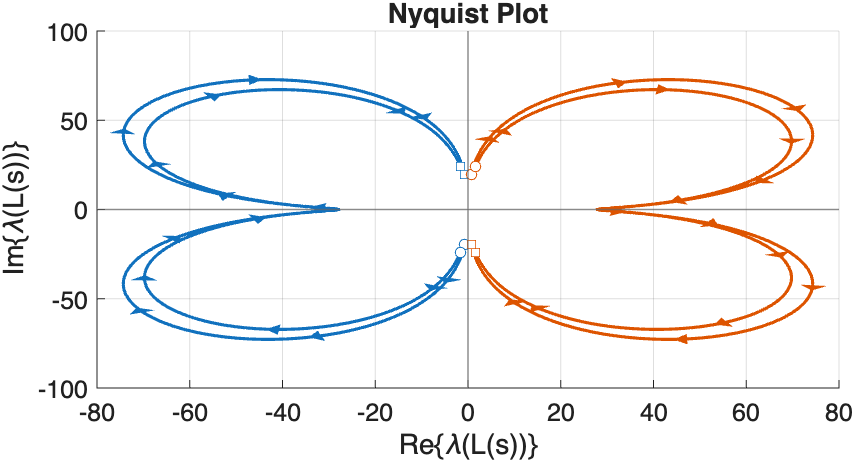}
    \caption{The Nyquist plot of $L_{G,i}(s)$ for the circuit shown in \cref{fig:tikz_series_impedance}
    \lhc{with $R = \SI{0.025}{\pu}$, $X = \SI{0.25}{\pu}$, $B = \SI{0.02}{\pu}$, \lhh{$\omega_0 = 2\pi \times \SI{50}{\hertz} = \SI{314}{\radian\per\second}$}, and with the power and voltage setpoints of the example in \cref{sec:model_instab_iv},}
    over the frequency band \SIrange{4e3}{5e3}{\radian\per\second} (near the resonant frequency).
    }
    \label{fig:plots_pq_unstable}
\end{figure}

Typically, \lhc{the plot of a strictly passive circuit $Z_i^p(s)$ along $\Gamma_N$} forms a large arc in the right half-plane near its resonant frequency \lhc{since the characteristic loci $a(s) \pm jb(s)$ lie in $\mathbb{C}_+$ by positive-realness. Under the transformation to $\check{Z}_i^p(s)$ and multiplication by $W_i$, however, \cref{lemma:pq_unstable} gives loci \eqref{eq:instab_pq_loci} where $\sqrt{a(s)^2 + b(s)^2}$ also lies in $\mathbb{C}_+$. Therefore, one branch of $\lambda_\pm$ always lies in $\bar{\mathbb{C}}_+$ and the other always lies in $\bar{\mathbb{C}}_-$.}

\lhc{For example, \cref{fig:plots_pq_unstable} shows the Nyquist plot of $L_{G,i}(s)=-\check{Z}_i^p(s)W_i$ (with negative sign to account for the positive-feedback interconnection in \eqref{eq:pq_from_iv}) along $\Gamma_N$ with the circuit shown in \cref{fig:tikz_series_impedance}.}
This shows two clockwise encirclements of the point $-1$ \lhc{(the loci are far from $-1$ at other frequencies)}, indicating that the $PQ$ transfer function of the $RLC$ element has two poles in $\mathbb{C}_+$. By direct calculation, we find that $G_i(s)$ has right half-plane poles at $s = 1.1 \times 10^3$ and $s = 1.6 \times 10^4$.

\begin{remark}
    \lhc{Real bus dynamics are more complex than the idealised models \eqref{eq:droop_tf} and \eqref{eq:gen_passive_compent}, so the full-order system should be checked to determine whether right half-plane poles arise when converting between the $PQ$ and $IV$ models. However, both results depend only on properties in a limited band: $s \rightarrow 0$ for \cref{lemma:iv_unstable}, and the resonant band of the series component for \cref{lemma:pq_unstable}. Therefore, the conclusions carry over to any bus whose dynamics are well approximated by \eqref{eq:droop_tf} at low frequencies, or for which \eqref{eq:Z_bus_approx_Z_passive} holds near resonance.}
\end{remark}
\begin{remark}
    \Cref{fig:tikz_series_impedance} represents many common configurations, such as a synchronous generator in series with a transformer, or a grid-forming or grid-following inverter with an $RLC$ filter. The dynamics of these high-frequency components are often ignored in power system stability studies involving $PQ$ models. However, \cref{lemma:pq_unstable} shows that these components should be treated with caution, as they may contribute unstable poles which must be stabilised by the interconnection.
\end{remark}

\subsection{Hybrid Model} \label{sec:model_hybrid}
To avoid the potential $IV$ and $PQ$ subsystem instabilities identified in \cref{sec:model_instab}, we propose here a \textit{hybrid} model that resembles the $PQ$ model at low frequencies and the $IV$ model at high frequencies using a frequency-dependent loop transformation.

First, for each $\nu_i \in \mathcal{V}_B$, let $F_{H,i}(s) \in \mathbf{RH}_\infty^{2\times2}$ be a high-pass filter with $F_{H,i}(0) = 0$, $\lim_{s \rightarrow \infty} F_{H,i}(s) = I$, and cutoff frequency $\omega_{c,i}$. Then define 
\begin{equation} \label{eq:hybrid_bus_def}
    \widetilde{G}_i(s) = (I+G_i(s)W_iF_{H,i}(s))^{-1}G_i(s),
\end{equation}
where $G_i(s)$ is the $PQ$ transfer function for bus $\nu_i \in \mathcal{V}_B$ and $W_i$ is defined in \eqref{eq:UW_defs}. \lhc{Since $F_{H,i}(j\omega) \approx 0$ at low frequencies,} $\widetilde{G}_i(j\omega) \approx G_i(j\omega)$ for $\omega < \omega_{c,i}$. Furthermore, we have that $U_i^\ddagger \widetilde{G}_i(j\omega) U_i^\dagger \approx Z_i(j\omega)$ when $\omega > \omega_{c,i}$, so we define 
\begin{equation} \label{eq:hybrid_impd_def}
    \widetilde{Z}_i(s) := U_i^\ddagger \widetilde{G}_i(s) U_i^\dagger.
\end{equation}

\lhc{Now, define the dual low-pass filter $F_{L,i}(s) := I - F_{H,i}(s)$ with cutoff frequency $\omega_{c,i}$} and let $F_L(s):= \oplus_{\nu_i \in \mathcal{V}_B} F_{L,i}(s)$ and $F_H(s):= \oplus_{\nu_i \in \mathcal{V}_B} F_{H,i}(s)$. Then, following \eqref{eq:N_def}, define
\begin{equation} \label{eq:hybrid_net_def}
    \begin{aligned} 
        \widetilde{N}_{H}(s) &:= N_{PQ}(s) - WF_H(s) \\
        &= U^\dagger Y_N(s) U^\ddagger + WF_L(s),
    \end{aligned}
\end{equation}
with $Y_N(s)$ defined in \eqref{eq:kcl_matrix_form}, and the matrices $U^\dagger$, $U^\ddagger$ and $W$ with blocks defined in \eqref{eq:UW_defs}. We see that $\widetilde{N}_{H}(j\omega) \approx N_{PQ}(j\omega)$ at low frequencies $\omega < \omega_{c}$, where $\omega_c := \min_{\nu_i \in \mathcal{V}_B} \omega_{c,i}$. At high frequencies, $F_L(j\omega)\rightarrow 0$ so $\widetilde{N}_{H}(j\omega) \approx U^\dagger Y_N(j\omega) U^\ddagger$. 
\lhc{To simplify the notation in future sections, we also define
\begin{equation} \label{eq:W_HL_def}
    \begin{gathered}
        \begin{aligned}
            W_{H,i} (s) := W_i F_{H,i}(s), && W_{L,i}(s) := W_i F_{L,i}(s),
        \end{aligned} \\
        \check{W}_{L,i}(s) := \tfrac{1}{{V_i^\star}^4} U_i^\dagger W_{L,i}(s) U_i^\ddagger,
    \end{gathered}
\end{equation}
and let $W_H(s) := \oplus_{\nu_i \in \mathcal{V}_B}W_{H,i}(s)$, with similar definitions for $W_L(s)$ and $\check{W}_L(s)$. We have $W_{L,i}(s) + W_{H,i}(s) = W_i$.}

Taking $\widetilde{G}_B(s) = \oplus_{\nu_i \in \mathcal{V}_B} \widetilde{G}_i(s)$, we obtain the negative feedback configuration $[\widetilde{G}_B(s), \widetilde{N}_{H}(s)]$ illustrated in \cref{fig:tikz_model_hyb}, \lhc{with \lhg{return ratio} $\widetilde{L}_H(s) := \widetilde{G}_B(s)\widetilde{N}_{H}(s)$}.

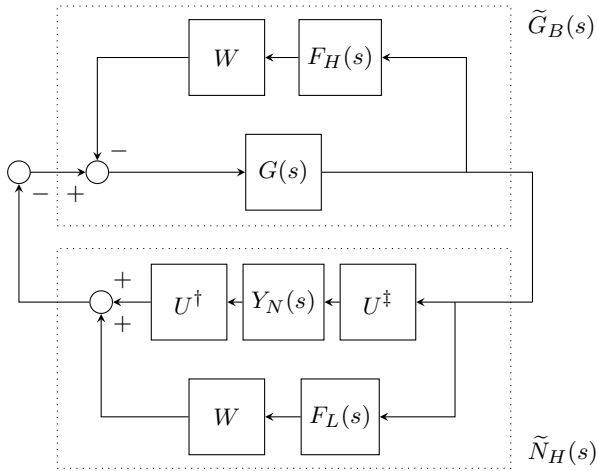
\begin{figure}[t]
\centering
\begin{tikzpicture}
\small

\node[draw,minimum size = 1cm]
    (sys_G) at (0,0)
    {$G(s)$};
\node[draw=black,
    below=of sys_G,yshift=0.3cm,
    minimum size = 1cm]
    (sys_Y) 
    {$Y_N(s)$};
    
\node[draw=black, minimum size = 1cm,
    above=of sys_G,yshift=-0.5cm,xshift=-0.75cm]
    (W_top) 
    {$W$};
\node[draw=black, minimum size = 1cm,
    above=of sys_G,yshift=-0.5cm,xshift=+0.75cm]
    (sys_FH) 
    {$F_H(s)$};
\node[draw=black, minimum size = 1cm,
    below=of sys_Y,yshift=0.5cm,xshift=-0.75cm]
    (W_btm) 
    {$W$};
\node[draw=black, minimum size = 1cm,
    below=of sys_Y,yshift=0.5cm,xshift=+0.75cm]
    (sys_FL) 
    {$F_L(s)$};

\node[draw=black, minimum size = 1cm,
    left=of sys_Y,xshift=0.8cm]
    (B) 
    {$U^\dagger$};
\node[draw=black, minimum size = 1cm,
    right=of sys_Y,xshift=-0.8cm]
    (Bt) 
    {$U^\ddagger$};

\node[draw, circle, minimum size=0.05cm]
    (neg) at (-3.5,0)
    {};
    \node[yshift=-5pt,xshift=5pt] at (neg.south east){$-$};
\node[draw, circle, minimum size=0.05cm,
    left=of sys_G,xshift=-0.8cm]
    (sum_top) 
    {};
    \node[yshift=5pt,xshift=5pt] at (sum_top.north east){$-$};
    \node[yshift=-5pt,xshift=-5pt] at (sum_top.south west){$+$};
\node[right= of sys_G,xshift=0.8cm] (right_top){};
\node[draw, circle, minimum size=0.05cm,
    left=of B,xshift=0.5cm]
    (sum_btm) 
    {};
    \node[yshift=5pt,xshift=5pt] at (sum_btm.north east){$+$};
    \node[yshift=-5pt,xshift=5pt] at (sum_btm.south east){$+$};
\node[right= of Bt,xshift=-0.6cm] (right_btm1){};
\node[right= of right_btm1,xshift=-0.2cm] (right_btm2){};

\node[draw,black,thin,dotted,anchor=north west,
    minimum width =6cm, minimum height =2.9cm]
    (box) at (-3,2.2) {};
    \node at (box.north east)[right=3pt,yshift=-7pt]{$\widetilde{G}_B(s)$};
\node[draw,black,thin,dotted,anchor=north west,
    minimum width =6cm, minimum height =2.9 cm]
    (box) at (-3,-1) {};
    \node at (box.south east)[right=3pt,yshift=7pt]{$\widetilde{N}_{H}(s)$};

\draw[-stealth] (neg.east) -- (sum_top.west);
\draw[-stealth] (sum_top.east) -- (sys_G.west);
\draw[-stealth] (right_top.center) |- (sys_FH.east);
\draw[-stealth] (sys_FH.west) -- (W_top.east);
\draw[-stealth] (W_top.west) -| (sum_top.north);
\draw[] (sys_G.east) -| (right_btm2.center);
\draw[-stealth] (right_btm2.center) -- (Bt.east);
\draw[-stealth] (right_btm1.center) |- (sys_FL.east);
\draw[-stealth] (sys_FL.west) -- (W_btm.east);
\draw[-stealth] (Bt.west) -- (sys_Y.east);
\draw[-stealth] (sys_Y.west) -- (B.east);
\draw[-stealth] (B.west) -- (sum_btm.east);
\draw[-stealth] (W_btm.west) -| (sum_btm.south);
\draw[-stealth] (sum_btm.west) -| (neg.south);

\end{tikzpicture}
\caption{Negative-feedback interconnection of the hybrid bus systems in $\widetilde{G}_B(s)$ and the hybrid network system $\widetilde{N}_{H}(s)$.}
\label{fig:tikz_model_hyb}
\end{figure}

Since $F_{H,i}(0) = 0$, $\widetilde{G}_i(s)$ and $\widetilde{Z}_i(s)$ \lhc{inherit \lhg{one} pole at the origin from $G_i(s)$}.
\lhc{As for right half-plane poles}, \cref{lemma:iv_unstable,lemma:pq_unstable} show that the dynamics that result in unstable poles in the $IV$ model arise due to low-frequency encirclements, while the encirclements that cause $PQ$ instabilities appear near the resonant frequency of the passive series component, which typically occurs at much higher frequencies. This indicates that a high-pass filter $F_{H,i}(s)$ can usually be found with a cutoff frequency between the dominant low- and high-frequency dynamics \lhg{such that each $\widetilde{G}_i(s)$ is stable (under \cref{def:stability}).}

Now, using \eqref{eq:hybrid_bus_def} and \eqref{eq:hybrid_net_def},
\begin{equation}
    \begin{aligned}
        (I + G_B(s)N_{PQ}(s))^{-1} = & \left[I+\widetilde{G}_B(s) \widetilde{N}_{H}(s)\right]^{-1}  \\
        & \times \left[I + G_B(s)WF_H(s)\right]^{-1} .
    \end{aligned}
\end{equation}
\lhg{If each $\widetilde{G}_i(s)$ is stable}, $(I + G_B(s)WF_H(s))^{-1}$ has no poles in $\bar{\mathbb{C}}_+ \setminus \{0\}$. Therefore, we have that $[G_B(s),N_{PQ}(s)]$ is stable under \cref{def:stability} if $[\widetilde{G}_B(s), \widetilde{N}_{H}(s)]$ is stable. 
By [Part I, \cref{pt1-thm:nyquist}], this can be concluded \lhc{either by direct analysis of the characteristic loci of $\widetilde{L}_H(s)$, or by verifying that \eqref{eq:nyquist_condition} holds for $\widetilde{L}_H(s)$}. In the following section, we show how the decentralised stability conditions derived in Part I can also be applied to $\widetilde{L}_H(s)$. 

\section{Stability Framework for Hybrid Model} \label{sec:stab}
In Part I, we introduced a decentralised stability framework to \lhc{certify} the Nyquist criterion \lhc{for either the admittance or impedance representation, using} sufficient quadratic constraints on the input--output properties of local subsystems \lhc{at three levels of locality}. 
\lhc{The extended-subsystem conditions, which deliver plug-and-play functionality, exploit the sparsity of the network and are only available in the impedance representation, which requires $Z_B(s)$ to be stable}. As \cref{sec:model_instab} \lhc{and [Part I, \cref{pt1-sec:case_study}]} show, the assumption \lhc{fails} in many common scenarios. 

In Part II, we instead propose that stability be certified using the hybrid grid model of \cref{sec:model_hybrid} so that no encirclements need to be accounted for when applying the Nyquist criterion. As in Part I, we use \eqref{eq:iqc_conditions} to derive decentralised conditions that are sufficient to conclude pointwise that at each $s \in \Gamma_N^{\,\epsilon}$,
\begin{equation} \label{eq:loci_condition}
    -1 \notin \sigma(\tau \widetilde{L}_H(s)), \quad \forall \tau \in [0,1].
\end{equation}
\lhc{In this section, we show that the framework and conditions derived in Part I can be directly extended to the hybrid model.}

\subsection{Bus-Level Conditions} \label{sec:stab_decentralised}
As in [Part I, \cref{pt1-lemma:iqc_decentralised}], we can exploit the block-diagonal structure of $\widetilde{G}_B(s)$ to yield decentralised constraints. 

\begin{lemma} \label{lemma:iqc_decentralised}
    Consider the hybrid grid system $[\widetilde{G}_B(s),\widetilde{N}_H(s)]$ \lhc{described in \cref{sec:model_hybrid} and} illustrated in \cref{fig:tikz_model_hyb}. Let $\widetilde{\Pi}(s) \in \mathbb{C}^{4N_B \times 4N_B}$ be a multiplier with the following block structure:
    \begin{equation} \label{eq:multiplier_decentralised}
        \widetilde{\Pi}(s) = \begin{bmatrix}
            \oplus_{\nu_i \in \mathcal{V}_B} \widetilde{\Pi}_{11}^i(s) & \oplus_{\nu_i \in \mathcal{V}_B} \widetilde{\Pi}_{12}^i(s)^\ast \\
            \oplus_{\nu_i \in \mathcal{V}_B} \widetilde{\Pi}_{12}^i(s) & \oplus_{\nu_i \in \mathcal{V}_B} \widetilde{\Pi}_{22}^i(s)
        \end{bmatrix},
    \end{equation}
    where for each $\nu_i \in \mathcal{V}_B$, $\widetilde{\Pi}_{11}^i(s) \in \mathbb{H}^2$, $\widetilde{\Pi}_{12}^i(s) \in \mathbb{C}^{2\times 2}$, and $\widetilde{\Pi}_{22}^i(s)\in \mathbb{H}^2$. Then, for a particular $s \in \Gamma_N^{\, \epsilon}$, condition \eqref{eq:loci_condition} holds if, for every $\tau \in [0,1]$, there exists a $\widetilde{\Pi}(s)$ of the form \eqref{eq:multiplier_decentralised} such that 
    \begin{subequations}
        \begin{equation} \label{eq:iqc_decentralised_network}
            \begin{bmatrix} 
                I_{2N_B} \\ -\tau \widetilde{N}_H(s)
            \end{bmatrix}^\ast \widetilde{\Pi}(s) \begin{bmatrix}
                I_{2N_B} \\ -\tau\widetilde{N}_H(s)
            \end{bmatrix} \leq 0,
        \end{equation}
        \begin{equation} \label{eq:iqc_decentralised_buses}
            \begin{bmatrix}
                \widetilde{G}_i(s) \\ I_2
            \end{bmatrix}^\ast \begin{bmatrix}
                \widetilde{\Pi}_{11}^i(s) & \ \widetilde{\Pi}_{12}^i(s)^\ast \\
                \widetilde{\Pi}_{12}^i(s) &  \widetilde{\Pi}_{22}^i(s)
            \end{bmatrix} \begin{bmatrix}
                \widetilde{G}_i(s) \\ I_2
            \end{bmatrix} > 0,
        \end{equation}
    \end{subequations}
    \lhh{where \eqref{eq:iqc_decentralised_buses} holds} at every $\nu_i \in \mathcal{V}_B$.
\end{lemma}
\begin{proof}
    This proof follows [Part I, \cref{pt1-lemma:iqc_decentralised}]. 
\end{proof}

As shown in Part I, in many cases, \lhg{multipliers leading to decentralised constraints} can be derived based on the properties of \lhg{the untransformed network $Y_N(s)$. The hybrid network $\widetilde{N}_H(s)$ is related to $Y_N(s)$ via \eqref{eq:hybrid_net_def}, so constraints on $Y_N(s)$ can be converted into inequalities of the form \eqref{eq:iqc_decentralised_network}. However, the homotopy parameter $\tau$ must be accounted for through the system transformation. In particular, [Part I, \cref{pt1-cor:iqc_convex}] (which allows constraints to be checked only at $\tau=1$) does not apply: because $W_i$ in \eqref{eq:UW_defs} is indefinite, the transformed blocks $\widetilde{\Pi}_{11}^i(s)$ and $\widetilde{\Pi}_{22}^i(s)$ in \eqref{eq:mult_blocks_transformed} lose the sign definiteness that result requires.}
The following corollary shows how \lhg{multipliers defined for the untransformed network extend to the hybrid model, where the resulting constraints need only be checked at the endpoints of the homotopy}. 

\begin{corollary} \label{cor:iqc_decentralised_imped}
    For a particular $s \in \Gamma_N^{\,\epsilon}$, let 
    \begin{equation} \label{eq:iqc_decentralised_network_Y}
        \begin{bmatrix} 
            I_{2N_B} \\ -Y_N(s)
        \end{bmatrix}^\ast \begin{bmatrix}
            \Pi_{11}(s) &  \Pi_{12}(s)^\ast \\ \Pi_{12}(s) & \Pi_{22}(s)
        \end{bmatrix} \begin{bmatrix}
            I_{2N_B} \\ -Y_N(s)
        \end{bmatrix} \leq 0,
    \end{equation}
    where $\Pi_{11}(s) := \oplus_{\nu_i \in \mathcal{V}_B} \Pi_{11}^i(s)$, $\Pi_{12}(s) := \oplus_{\nu_i \in \mathcal{V}_B} \Pi_{12}^i(s)$, and $\Pi_{22}(s) := \oplus_{\nu_i \in \mathcal{V}_B} \Pi_{22}^i(s)$, and for each $\nu_i \in \mathcal{V}_B$, $\Pi_{11}^i(s) = \Pi_{11}^i(s)^\ast \leq 0$ and $\Pi_{22}^i(s) = \Pi_{22}^i(s)^\ast \geq 0$. 
    \lhc{Let
    \begin{equation} \label{eq:M_def}
        M_i(s) = \frac{1}{4} \check{W}_{L,i}(s)^\ast \Pi_{22}^i(s) \check{W}_{L,i}(s),
    \end{equation}
    where $\check{W}_{L,i}(s)$ is defined in \eqref{eq:W_HL_def}.}
    Then, condition \eqref{eq:loci_condition} holds if, for every $\nu_i \in \mathcal{V}_B$, both the following conditions are satisfied:
    \begin{subequations} \label{eq:iqc_decentralised_buses_Z}
        \begin{gather}
            \begin{bmatrix}
                \widetilde{Z}_i(s)  \\ I_2
            \end{bmatrix}^\ast \begin{bmatrix}
                \Pi_{11}^i(s) \lhc{-M_i(s)} & \Pi_{12}^i(s)^\ast \\
                \Pi_{12}^i(s) &  \Pi_{22}^i(s)
            \end{bmatrix} \begin{bmatrix}
                \widetilde{Z}_i(s) \\ I_2
            \end{bmatrix} > 0, \label{eq:iqc_decentralised_buses_Z_tilde} \\
            \begin{bmatrix}
                Z_i(s) \\ I_2
            \end{bmatrix}^\ast \begin{bmatrix}
               \Pi_{11}^i(s) \lhc{-M_i(s)} & \Pi_{12}^i(s)^\ast \\
                \Pi_{12}^i(s) &  \Pi_{22}^i(s)
            \end{bmatrix} \begin{bmatrix}
                Z_i(s) \\ I_2
            \end{bmatrix} > 0, \label{eq:iqc_decentralised_buses_Z_no_tilde} 
        \end{gather}
    \end{subequations}
    where $\widetilde{Z}_i(s)$ and $Z_i(s)$ are the \lhc{hybrid and $IV$} bus impedances defined in \eqref{eq:hybrid_impd_def} and \eqref{eq:bus_iv_relationship}, respectively.
\end{corollary}
\begin{proof}
    The corollary is proved in Appendix \ref{app:iqc_decentralised_imped_proof}.
\end{proof}
Therefore, all the conditions in \cref{pt1-cor:dw_shell,pt1-cor:pos_real,pt1-cor:inf_norm,pt1-cor:fixed_net,pt1-cor:conic} of Part I, including any graphical interpretations, can be directly applied to the hybrid bus system. 
\begin{remark}
    \lhc{The two conditions in \eqref{eq:iqc_decentralised_buses_Z} are the endpoints of the homotopy: as $\tau$ varies over $[0,1]$, the transformation~\eqref{eq:hybrid_bus_def} is continuously undone, so the system moves from the hybrid impedance $\widetilde{Z}_i(s)$ at $\tau = 0$ to the $IV$ impedance $Z_i(s)$ at $\tau = 1$. The \lhg{multiplier modification} $M_i(s)$ allows us to only check these endpoints. It vanishes for conditions where $\Pi_{22}^i(s) = 0$ \lhg{(e.g.,the positive-real condition)}, and, since $\check{W}_{L,i}(s) \rightarrow 0$ above $\omega_{c,i}$, it is negligible in the high-frequency ranges where conditions with $\Pi_{22}^i(s) \neq 0$ are typically applied, so \eqref{eq:iqc_decentralised_buses_Z} approaches its Part I counterpart there.}
\end{remark}

\subsection{Extended Subsystem-Level Conditions} \label{sec:stab_extended}
\lhc{As the hybrid network $\widetilde{N}_H(s)$ retains the sparsity of $Y_N(s)$}, the extended-subsystem formulation of [Part I, \cref{pt1-sec:stab_extended}] can similarly be derived for the hybrid model. In this case, each extended subsystem contains the dynamics of a single bus and all the lines, loads, and the loop transformation associated with that bus.

In particular, we define the extended graph $(\mathcal{V}_B,\widetilde{\mathcal{E}}_N)$ for the hybrid model as follows: 
let 
\begin{equation} \label{eq:YE_def}
    \widetilde{Y}_E(s) := \begin{bmatrix}
        Y_P(s) & \\ & Y_0(s) + \check{W}_{L}(s)
    \end{bmatrix},    
\end{equation}
where $Y_P(s)$ and $Y_0(s)$ are as in \eqref{eq:network_admittance}, \lhh{and $\check{W}_L(s):= \oplus_{\nu_i \in \mathcal{V}_B} \check{W}_{L,i}(s)$, with $\check{W}_{L,i}(s)$ defined in \eqref{eq:W_HL_def}.} Therefore, the extended incidence matrix in this case is given by
$
    \widetilde{B}_E := \begin{bmatrix}
        B_N & I_{N_B} 
    \end{bmatrix},
$
where $B_N$ is the incidence matrix defined \eqref{eq:incidence_matrix},
meaning our extended graph contains $N_B$ nodes given by the set $\mathcal{V}_B$ and $N_E:=N_P + N_B$ edges given by the set $\widetilde{\mathcal{E}}_N$. Let $\widetilde{\mathcal{B}}_E = \widetilde{B}_E \otimes I_2$. \lhc{Note that $\widetilde{\mathcal{E}}_N$ only differs from $\mathcal{E}_N$ in that there is a shunt edge at every bus in $\widetilde{\mathcal{E}}_N$.}

With the extended graph defined, for each $\nu_i \in \mathcal{V}_B$, we can similarly derive selection matrices $\widetilde{\mathcal{K}}_i \in \mathbb{R}^{2N_E \times 2\tilde{m}_i}$ (see [Part I, \cref{pt1-sec:stab_extended}]), where $\tilde{m}_i = p_i + 1$, with $p_i$ equal to the number of power lines connected to $\nu_i \in \mathcal{V}_B$. We then let 
\begin{equation} \label{eq:K_mat_def}
    \widetilde{\mathcal{K}}:= \begin{bmatrix}
        \widetilde{\mathcal{K}}_1 & \widetilde{\mathcal{K}}_2 & \cdots & \widetilde{\mathcal{K}}_{N_B}
    \end{bmatrix}.
\end{equation}
Finally, for each $\nu_i \in \mathcal{V}_B$, define 
\begin{equation} \label{eq:E_sys_def}
    \widetilde{E}_i(s) := \widetilde{Z}_{E,i}(s) \widetilde{Y}_{E,i}(s),
\end{equation}
where 
\begin{subequations} \label{eq:ZE_LE_def}
    \begin{align} 
        \widetilde{Z}_{E,i}(s) &:= \widetilde{\mathcal{K}}_i^{\ T} (\widetilde{\mathcal{B}}_E^{i\bullet})^T \widetilde{Z}_i(s) \widetilde{\mathcal{B}}_E^{i\bullet} \widetilde{\mathcal{K}}_i, \\
        \widetilde{Y}_{E,i}(s) & := \widetilde{\mathcal{K}}_i^{\ T} \widetilde{Y}_E(s) \widetilde{\mathcal{K}}_i,
     \end{align}
\end{subequations}
with $\widetilde{Z}_i(s)$ and $\widetilde{Y}_E(s)$ defined in \eqref{eq:hybrid_impd_def} and \eqref{eq:YE_def}, respectively, and $\widetilde{\mathcal{B}}_E^{i\bullet}$ is the $i^{\mathrm{th}}$ $2 \times 2N_E$ block-row of $\widetilde{\mathcal{B}}_E$. Then let $\widetilde{E}(s) = \oplus_{\nu_i \in \mathcal{V}_B} \widetilde{E}_i(s)$.

Now, using \eqref{eq:network_admittance}, \eqref{eq:U_squared} and \eqref{eq:hybrid_net_def}, (and omitting the argument $s$ for convenience), we have
\begin{equation*}
    \begin{aligned}
        \widetilde{L}_H 
        & = \widetilde{G}_B\widetilde{N}_H 
        = \widetilde{G}_B (U^\dagger \mathcal{B}_N Y_P \mathcal{B}_N^T U^\ddagger + U^\dagger Y_0 U^\ddagger + W F_L) \\ 
        & = \widetilde{G}_B U^\dagger ( \mathcal{B}_N Y_P \mathcal{B}_N^T  +  Y_0  + \check{W}_L) U^\ddagger \\
        & = {V^\star}^{-2} U^\ddagger \widetilde{Z}_B \widetilde{\mathcal{B}}_E \widetilde{Y}_E \widetilde{\mathcal{B}}_E^{\ T} U^\ddagger,
    \end{aligned}
\end{equation*}
where $\widetilde{Z}_B(s) = \oplus_{\nu_i \in \mathcal{V}_B} \widetilde{Z}_i(s)$ with $\widetilde{Z}_i(s)$ given in \eqref{eq:hybrid_impd_def}, and $V^\star:= (\oplus_{\nu_i \in \mathcal{V}_B} V_i^\star)\otimes I_2$. Therefore, the non-zero eigenvalues of $\widetilde{L}_H$ equal those of $\widetilde{Z}_B \widetilde{\mathcal{B}}_E \widetilde{Y}_E \widetilde{\mathcal{B}}_E^{\ T}$. 
Following a similar approach to the proof of [Part I, \cref{pt1-lemma:reform_extended}], we can then show that, for a particular $s \in \Gamma_N^{\,\epsilon}$, the condition \eqref{eq:loci_condition} holds if 
\begin{equation} \label{eq:loci_condition_extended}
    -1 \notin \sigma(\tau  \widetilde{E}(s) \widetilde{\mathcal{K}}^{\ T} \widetilde{\mathcal{K}}), \quad \forall \tau \in [0,1]. 
\end{equation}
Consequently, all the results of [Part I, \cref{pt1-sec:stab_extended}] directly carry over to the hybrid model, so we re-state the following version of [Part I, \cref{pt1-lemma:iqc_ext}]. 

\begin{lemma} \label{lemma:iqc_ext}
    Consider the reformulated grid system $[\widetilde{E}(s), \widetilde{\mathcal{K}}^{\ T} \widetilde{\mathcal{K}}]$, defined above. For a particular $s \in \Gamma_N^{\,\epsilon}$ and for each $(\nu_i,\nu_j) \in \widetilde{\mathcal{E}}_N$, let $\widetilde{\pi}_{12}^{ij}(s) \in \mathbb{C}^{2\times2}$ and $\widetilde{\pi}_{22}^{ij}(s) \in \mathbb{H}^2$, with $\widetilde{\pi}_{22}^{ij}(s) \geq 0$, be multiplier blocks such that the following matrix inequality holds: 
    \begin{equation} \label{eq:iqc_ext_condition}
        -\widetilde{\pi}_{12}^{ij}(s) - \widetilde{\pi}_{12}^{ij}(s)^\ast + n_{ij}\widetilde{\pi}_{22}^{ij}(s) \leq 0,
    \end{equation}
    where $n_{ij} = 2$ if $(\nu_i,\nu_j)\in \widetilde{\mathcal{E}}_N$ is a power line, and $n_{ij} = 1$ if $(\nu_i,\nu_j) \in \widetilde{\mathcal{E}}_N$ contains a load or shunt together with a loop transformation. Let $\widetilde{\pi}_{12}(s) = \oplus_{(\nu_i,\nu_j) \in \widetilde{\mathcal{E}}_N} \widetilde{\pi}_{12}^{ij}(s)$ and $\widetilde{\pi}_{22}(s) = \oplus_{(\nu_i,\nu_j) \in \widetilde{\mathcal{E}}_N} \widetilde{\pi}_{22}^{ij}(s)$. Then, for each $\nu_i \in \mathcal{V}_B$, take
    \begin{subequations}
        \begin{align} 
                \widetilde{\Pi}_{12}^i(s):= \widetilde{\mathcal{K}}_i^{\ T} \widetilde{\pi}_{12}(s) \widetilde{\mathcal{K}}_i, \label{eq:iqc_ext_structure_12} \\  \widetilde{\Pi}_{22}^i(s):= \widetilde{\mathcal{K}}_i^{\ T} \widetilde{\pi}_{22}(s) \widetilde{\mathcal{K}}_i,
        \end{align}
    \end{subequations}
    i.e., the blocks of $\widetilde{\Pi}_{12}^i(s)$ and $\widetilde{\Pi}_{22}^i(s)$ are selected to include the multiplier terms associated with the edges connected to $\nu_i \in \mathcal{V}_B$. 
    Then, \eqref{eq:loci_condition_extended} holds if, for every $\nu_i \in \mathcal{V}_B$,
    \begin{equation} \label{eq:iqc_ext_local}
        \begin{bmatrix}
            \widetilde{E}_i(s) \\ I_{2\tilde{m}_i}
        \end{bmatrix}^\ast \begin{bmatrix}
            0 & \ \widetilde{\Pi}_{12}^i(s)^\ast \\
            \widetilde{\Pi}_{12}^i(s) &  \widetilde{\Pi}_{22}^i(s)
        \end{bmatrix} \begin{bmatrix}
            \widetilde{E}_i(s) \\ I_{2\tilde{m}_i}
        \end{bmatrix} > 0.
    \end{equation} 
\end{lemma}
\begin{proof}
    The proof follows [Part I, \cref{pt1-lemma:iqc_ext}]. 
\end{proof}

Furthermore, the conditions in \cref{pt1-cor:ext_stability_simp,pt1-cor:ext_stability_nr,pt1-cor:ext_stability} of Part I give specific examples of multiplier blocks that satisfy \eqref{eq:iqc_ext_condition} and a corresponding graphical interpretation. As the constraint \eqref{eq:iqc_ext_condition} is local to each edge, use of \cref{lemma:iqc_ext} to certify stability can be used to form stability certifications that allow for plug-and-play functionality. 

\section{Low-Frequency Nyquist Analysis} \label{sec:low_freq}
The framework presented in Part I and \cref{sec:stab} certifies \eqref{eq:loci_condition} \lhc{by decentralised means. However, the case study in Part I showed that no decentralised condition was feasible at low frequencies without a stiff voltage source, because the characteristic loci approach $-1$ in this range. In this section, we analyse the Nyquist plot of $\widetilde{L}_H(s)$ directly and show that the characteristic polynomial decouples near $s=0$ into a factor governing frequency dynamics and a factor governing voltage dynamics, recovering the classical $P$--$\delta$ / $Q$--$V$ separation of power system analysis in a Nyquist setting. The first factor produces unbounded loci controlled by a decentralised condition on local setpoints, while the second
\lhg{produces bounded loci governed by a matrix closely related to}
the \textit{reduced power-flow Jacobian} that operators already constrain during dispatch.}

For each $\nu_i \in \mathcal{V}_B$, write the entries of $\widetilde{G}_i(s)$ as in \eqref{eq:pq_transfer_function}, with $\widetilde{G}_{\omega p,i}(s)$, $\widetilde{G}_{\omega q,i}(s)$, $\widetilde{G}_{vp,i}(s)$ and $\widetilde{G}_{vq,i}(s)$ having no poles in $\bar{\mathbb{C}}_+$, and let $\widetilde{G}_{\omega p}(s) := \oplus_{\nu_i \in \mathcal{V}_B} \widetilde{G}_{\omega p,i}(s)$, with $\widetilde{G}_{\omega q}(s)$, $\widetilde{G}_{vp}(s)$ and $\widetilde{G}_{vq}(s)$ defined analogously. Permuting the
rows and columns of $\widetilde{G}_B(s)$ and $\widetilde{N}_H(s)$ to group the frequency and voltage channels gives
\begin{equation*}
\setlength{\arraycolsep}{0.05cm}
    \widetilde{G}_B^P(s) 
    = \begin{bmatrix}
        \frac{1}{s} \widetilde{G}_{\omega p}(s) & \frac{1}{s} \widetilde{G}_{\omega q}(s) \\ \widetilde{G}_{vp}(s) & \widetilde{G}_{vq}(s)
    \end{bmatrix}\!, \, 
    \widetilde{N}_H^P(s) = \begin{bmatrix}
        \widetilde{N}_{p\delta}(s) & \widetilde{N}_{pv}(s) \\ \widetilde{N}_{q\delta}(s) & \widetilde{N}_{qv}(s)
    \end{bmatrix}\!,
\end{equation*}
with permuted \lhg{return ratio} 
\begin{equation} \label{eq:open_loop_permuted}
    \widetilde{L}_H^P(s) 
    := \widetilde{G}_B^P(s) \widetilde{N}_H^P(s) \\
    = \begin{bmatrix}
    \frac{1}{s} \widetilde{L}_{\omega \delta}(s) & \frac{1}{s} \widetilde{L}_{\omega v}(s) \\ \widetilde{L}_{v \delta}(s) & \widetilde{L}_{vv}(s)
    \end{bmatrix},
\end{equation}
where 
\begin{subequations} \label{eq:open_loop_permuted_blocks}
    \begin{align}
        \widetilde{L}_{\omega\delta}(s) &:= \widetilde{G}_{\omega p}(s) \widetilde{N}_{p\delta}(s) + \widetilde{G}_{\omega q}(s) \widetilde{N}_{q\delta}(s), \label{eq:open_loop_permuted_blocks_Lwd} \\ 
        \widetilde{L}_{\omega v}(s) &:= \widetilde{G}_{\omega p}(s) \widetilde{N}_{pv}(s) + \widetilde{G}_{\omega q}(s) \widetilde{N}_{qv}(s), \\
        \widetilde{L}_{v \delta}(s) &:= \widetilde{G}_{v p} (s)\widetilde{N}_{p\delta}(s) + \widetilde{G}_{v q}(s) \widetilde{N}_{q\delta}(s), \\
        \widetilde{L}_{vv}(s) &:= \widetilde{G}_{v p}(s) \widetilde{N}_{pv}(s) + \widetilde{G}_{v q}(s) \widetilde{N}_{qv}(s).
    \end{align}    
\end{subequations}
\lhc{The permutation is a similarity transformation, so $\widetilde{L}_H^P(s)$ and $\widetilde{L}_H(s)$ have the same spectrum. We now analyse the behaviour of that spectrum as $s \rightarrow 0$.}

\subsection{Low-Frequency Network Model} \label{sec:low_freq_model}
\lhc{We first show that, at low frequencies, the hybrid network system reduces to the power-flow Jacobian of the grid \cite[Section 14.3]{kundur_PowerSystem_22}.}
As discussed in \cref{sec:model_hybrid}, for $\omega< \omega_c$, $F_L(j\omega) \approx I$, so near $s=0$, $\widetilde{N}_H(s) \approx N_{PQ}(s) = U^\dagger Y_N(s) U^\ddagger + W$ by \eqref{eq:hybrid_net_def}. We now make the following assumption on the low-frequency network dynamics. 
\begin{assumption} \label{assump:low_freq_simp}
    \lhc{Each $2\times2$ block of $Y_N(0)$ in \eqref{eq:network_admittance} takes the form $G I_2 + B J$ with $G, B \in \mathbb{R}$.} Furthermore, there exists $\omega_l \in (0,\omega_c]$ such that $Y_N(s) \approx Y_N(0)$ for all $\lvert s \rvert \leq \omega_l$.
\end{assumption}
\lhc{Any network composed of admittances of $RLC$ elements, such as [Part I, Equation \eqref{pt1-eq:line_admittance}] and [Part I, Equation \eqref{pt1-eq:shunt_admittance}], takes the required form.} 
Therefore, near $s = 0$, $\widetilde{N}_H(s) \approx \widetilde{N}_H(0)$. 
\lhc{Evaluating \eqref{eq:powers} at the equilibrium, with ${v_i^{DQ}}^\star = V_i^\star[\cos\delta_i^\star,\ \sin\delta_i^\star]^T$ and ${i_B^{DQ}}^\star = Y_N(0){v_B^{DQ}}^\star$ from \eqref{eq:kcl_matrix_form}, yields the familiar power-flow equations:}
\begin{subequations} \label{eq:power_flow}
    \begin{align}
        &\begin{aligned}
            & P_i^\star = \underbrace{{V_i^\star}^2 \left(G_i + \textstyle\sum_{j \neq i,0}^{N_B} G_{ij} \right)}_{P_{ii}^\star} \\
            &\; \!+\! \sum_{j \neq i,0}^{N_B} \underbrace{V_i^\star V_j^\star\left[-  G_{ij} \cos(\delta_j^\star \!-\! \delta_i^\star) \!+\! B_{ij} \sin(\delta_j^\star \!-\! \delta_i^\star) \right]}_{P_{ij}^
            \star},
        \end{aligned} \label{eq:power_flow_p}\\
        &\begin{aligned}
            Q_i^\star &= \underbrace{-{V_i^\star}^2 \left(B_i + \textstyle\sum_{j \neq i,0}^{N_B} B_{ij} \right)}_{Q_{ii}^\star} \\
            & \!+\! \sum_{j \neq i,0}^{N_B} \underbrace{V_i^\star V_j^\star\left[B_{ij} \cos(\delta_j^\star \!-\! \delta_i^\star) \!+\! G_{ij} \sin(\delta_j^\star \!-\! \delta_i^\star) \right]}_{Q_{ij}^
            \star}.
        \end{aligned} \label{eq:power_flow_q}
    \end{align}
\end{subequations}
Note that in general $P_{ij}^\star \neq P_{ji}^\star$ and $Q_{ij}^\star \neq Q_{ji}^\star$. Now, through some similar algebra, we use \eqref{eq:pq_deriv_1} to calculate 
\begin{equation*}
    \begin{aligned}
        \begin{bmatrix}
            \Delta P_i(0) \\ \Delta Q_i(0)
         \end{bmatrix} 
    & = \left( \begin{bmatrix}
             -\sum_{j \neq i,0}^{N_B} Q_{ij}^\star & P_i^\star + P_{ii}^\star \\ \sum_{j \neq i,0}^{N_B} P_{ij}^\star & Q_i^\star + Q_{ii}^\star
         \end{bmatrix}  \right)  \begin{bmatrix}
             \Delta \delta_i(0) \\ \frac{\Delta V_i(0)}{V_i^\star} 
         \end{bmatrix} \\
         &\qquad\qquad+ \sum_{j \neq i,0}^{N_B} \begin{bmatrix}
             Q_{ij}^\star & P_{ij}^\star \\ - P_{ij}^\star & Q_{ij}^\star
         \end{bmatrix} \begin{bmatrix}
             \Delta \delta_j(0) \\ \frac{\Delta V_j(0)}{V_j^\star} 
         \end{bmatrix},
    \end{aligned}
\end{equation*}
which gives us the entries of $\widetilde{N}_H(s)$ near $s = 0$. Performing the row and column permutation required to form $\widetilde{N}_H^P(s)$ gives, for all $s \in \Gamma_N^{\, \epsilon}$ with $|s|\leq \omega_l$,
\begin{equation} \label{eq:low_freq_net}
    \widetilde{N}_H^P(s) \approx \begin{bmatrix}
        J_{p\delta} & J_{pv} \\ J_{q\delta} & J_{qv}
    \end{bmatrix},
\end{equation}
i.e., $\widetilde{N}_H^P(s)$ reduces to the power-flow Jacobian \cite[Section 14.3]{kundur_PowerSystem_22}, where the entries of each block are 
\begin{equation} \label{eq:low_freq_net_blocks}
    \begin{array}{c|cc}
         & i=j & i\neq j \\ \hline \rule{0pt}{12pt}
        J_{p\delta}^{(i,j)} & -\sum_{k \neq i,0}^{N_B} Q_{ik}^\star, & Q_{ij}^\star, \\[4pt]
        J_{pv}^{(i,j)}      & P_i^\star + P_{ii}^\star,       & P_{ij}^\star,  \\[4pt]
        J_{q\delta}^{(i,j)} & \sum_{k \neq i,0}^{N_B} P_{ik}^\star,  & -P_{ij}^\star, \\[4pt]
        J_{qv}^{(i,j)}      & Q_i^\star + Q_{ii}^\star,       & Q_{ij}^\star.
    \end{array}
\end{equation}
\lhc{Each block in \eqref{eq:low_freq_net_blocks} is in general non-symmetric, but $J_{p\delta}$ and $J_{q\delta}$ have zero row sums. Hence, for $\lvert s\rvert \leq \omega_l$,}
\begin{equation} \label{eq:null_identities}
    \lhc{\widetilde{L}_{\omega\delta}(s)\mathbf{1} = 0, \qquad
    \widetilde{L}_{v\delta}(s)\mathbf{1} = 0,}
\end{equation}
\lhc{since both blocks are combinations of $J_{p\delta}$ and $J_{q\delta}$ by \eqref{eq:open_loop_permuted_blocks}. }

\subsection{Decoupling of the Characteristic Loci} \label{sec:low_freq_loci}
We now simplify the characteristic polynomial of $\widetilde{L}_H^P(s)$ near $s = 0$ using the low-frequency network dynamics \eqref{eq:low_freq_net}. 
\lhc{By \eqref{eq:null_identities}, $\widetilde{L}_{\omega\delta}(s)$ is singular for $\lvert s \rvert \leq \omega_l$, so the derivation below requires care.}
\begin{assumption} \label{assump:simple_eigenvalue}
    For $|s| \leq \omega_l$, the zero eigenvalue of $\widetilde{L}_{\omega\delta}(s)$ \lhh{established} \lhc{in \eqref{eq:null_identities}} is simple.
\end{assumption}

\lhc{This assumption is mild under normal operating conditions due to the Laplacian-like structure of $J_{p\delta}$ and $J_{q\delta}$ in \eqref{eq:low_freq_net_blocks} for a connected network (\cref{assump:connected})}.
Consequently, $\mathrm{null}(\widetilde{L}_{\omega\delta}(s)) = \mathrm{span}(\mathbf{1})$, so $\widetilde{L}_{\omega\delta}(s)$ is a matrix of index $1$ and its group inverse $\widetilde{L}_{\omega\delta}(s)^\#$ exists (see \cref{def:group_inv}).

\begin{lemma} \label{lemma:char_poly}
    Consider the \lhg{return ratio} $\widetilde{L}_H^P(s)$ given by \eqref{eq:open_loop_permuted} under \cref{assump:low_freq_simp,assump:simple_eigenvalue}. For $s \in \Gamma_N^{\, \epsilon}$, with $|s| \leq \omega_l$, the characteristic polynomial $\det(\widetilde{L}_H^P(s) - \lambda I) =0$ is given by
    \begin{equation} \label{eq:char_poly_decoupled}
        \det(\tfrac{1}{s} \widetilde{L}_{\omega\delta}(s) - \lambda I)\times \det( \widetilde{L}_R^P(s) - \lambda I + \mathcal{O}(s\lambda)) = 0, 
    \end{equation}
where
\begin{equation} \label{eq:open_loop_reduced_near_0}
    \widetilde{L}_R^P(s) = \widetilde{L}_{vv}(s) - \widetilde{L}_{v\delta}(s)\widetilde{L}_{\omega\delta}(s)^\#  \widetilde{L}_{\omega v}(s),
\end{equation}
and $\widetilde{L}_{\omega\delta}(s)^\#$ is the group inverse of $\widetilde{L}_{\omega\delta}(s)$. 
\end{lemma}
\begin{proof}
    The proof is given in Appendix \ref{app:char_poly_proof}.
\end{proof}

\lhc{This decoupling produces the classical $P$--$\delta$ / $Q$--$V$ separation of power system analysis in a Nyquist setting. The first determinant in \eqref{eq:char_poly_decoupled} is exact and its roots give $N_B$ of the characteristic loci directly. Of these, one branch is identically zero by \eqref{eq:null_identities}, while the other $N_B-1$ take the form $\lambda = \mu/s$ with $\mu \in \sigma(\widetilde{L}_{\omega\delta}(s))\setminus\{0\}$ and are therefore unbounded as $s \rightarrow 0$. The second determinant gives the remaining $N_B$ loci and depends on $\lambda$ through the $\mathcal{O}(s\lambda)$ term. For the branches it describes, $\lambda$ remains bounded as $s \rightarrow 0$, so this term vanishes and the second factor reduces to the eigenvalue problem for $\widetilde{L}_R^P(s)$ in \eqref{eq:open_loop_reduced_near_0}. The characteristic loci near $s = 0$ are therefore the union of two sets, characterised by the following corollary.} 

\begin{corollary} \label{cor:low_freq}
    \lhc{Let \cref{assump:low_freq_simp,assump:simple_eigenvalue} hold. Then, for $s \in \Gamma_N^{\, \epsilon}$ with $|s| \leq \omega_l$, \eqref{eq:loci_condition} holds if}
    \begin{subequations} \label{eq:low_freq_conditions}
        \begin{align} 
            &\Re(\mu) > 0, && \text{for all } \mu \in \sigma(\widetilde{L}_{\omega\delta}(s))\setminus\{0\},  \label{eq:unbounded_cond} \\
            &\lambda \notin (-\infty,-1], && \text{for all } \lambda \in \sigma(\widetilde{L}_R^P(s)). \label{eq:bounded_cond}
        \end{align}
    \end{subequations}
\end{corollary}
\begin{proof}
    The proof is given in Appendix \ref{app:low_freq_proof}. 
\end{proof}

Conditions \eqref{eq:low_freq_conditions} can each be verified using Gershgorin's circle theorem applied to the rows of their respective matrices \cite{horn_MatrixAnalysis_85}. As each row of $\widetilde{L}_{\omega\delta}(s)$ is sparse and depends only on parameters and setpoints local to $\nu_i \in \mathcal{V}_B$ and its neighbours, \lhc{this constitutes a decentralised check for \eqref{eq:unbounded_cond}}. \lhc{For \eqref{eq:bounded_cond}, however, the group inverse in \eqref{eq:open_loop_reduced_near_0} destroys the sparsity of $\widetilde{L}_{\omega\delta}(s)$, so the rows of $\widetilde{L}_R^P(s)$ are not local to individual buses. This condition therefore requires centralised analysis.}

\subsection{Behaviour at Typical Operating Points} \label{sec:low_freq_behaviour}
\lhc{We now argue that both conditions \eqref{eq:low_freq_conditions} are typically satisfied under standard network configurations and dispatch, so the Nyquist plot at low frequencies is usually well behaved.}

For the unbounded loci, a predominantly frequency-to-active-power droop-controlled bus has $\widetilde{G}_{\omega q}(s) \approx 0$ at low frequencies, so $\widetilde{L}_{\omega\delta}(s) \approx \widetilde{G}_{\omega p}(s)J_{p\delta}$. At normal operating points, \lhc{transmission branches have} $B_{ij} < 0$ with $\lvert B_{ij}\rvert \gg \lvert G_{ij}\rvert$ and small angle differences, so \eqref{eq:power_flow_q} gives $Q_{ij}^\star < 0$ and, by \eqref{eq:low_freq_net_blocks}, $J_{p\delta}$ has positive diagonal entries, non-positive off-diagonal entries and zero row sums. Its non-zero eigenvalues are therefore in the right half-plane. Since $\widetilde{G}_{\omega p,i}(s) \approx k_{\omega,i} > 0$ near $s = 0$, this structure is retained by $\widetilde{L}_{\omega\delta}(s)$ and \eqref{eq:unbounded_cond} is satisfied.

For the bounded loci, consider the simplified case $\widetilde{G}_{\omega q}(s) \approx \widetilde{G}_{vp}(s) \approx 0$ and $\widetilde{G}_{\omega p}(s) \approx \tilde{g}_{\omega p}(s)I_{N_B}$, i.e., approximately homogeneous frequency-to-active-power dynamics across the network. Then \eqref{eq:open_loop_reduced_near_0} becomes
\begin{equation} \label{eq:reduced_jacobian}
    \widetilde{L}_R^P(s) \approx \widetilde{G}_{vq}(s)
    \left( J_{qv} - J_{q\delta}J_{p\delta}^{\#}J_{pv} \right),
\end{equation}
in which the bracketed matrix is closely related to the \textit{reduced power-flow Jacobian} $J_R$ \cite[Section 14.3]{kundur_PowerSystem_22}, \cite{gao_VoltageStability_92}. Grid operators already constrain $J_R$ during generator dispatch to keep its eigenvalues in the right half-plane, since left half-plane eigenvalues \lhc{indicate proximity to voltage collapse}. Scaled by an approximately homogeneous $\widetilde{G}_{vq}(s)$, the eigenvalues of $\widetilde{L}_R^P(s)$ therefore also lie in the right half-plane, safely clear of~$-1$. \lhg{For heterogeneous $\widetilde{G}_{\omega p}(s)$ and $\widetilde{G}_{vq}(s)$, the loci of $\widetilde{L}_R^P(s)$ often obey \eqref{eq:bounded_cond} by the same heuristic, as shown in the case study in \cref{sec:case_study}; for more complex cases, the conditions \eqref{eq:low_freq_conditions} can be incorporated into power-flow analysis during generator dispatch.}

Consequently, the conditions \eqref{eq:low_freq_conditions} can be omitted for many standard network configurations, and decentralised grid codes can concentrate on the higher frequency ranges where problematic electromagnetic interactions arise. 

\section{Combined Stability Criteria} \label{sec:comb}
The results of the previous sections are now combined to produce a set of criteria that can be used to certify the small-signal stability of the hybrid grid system $[\widetilde{G}_B(s),\widetilde{N}_H(s)]$. 

\begin{proposition} \label{prop:stability}
    Consider the hybrid grid system $[\widetilde{G}_B(s),\widetilde{N}_H(s)]$ described in \cref{sec:model_hybrid} and illustrated in \cref{fig:tikz_model_hyb}, \lhg{such that $\widetilde{G}_i(s) \in \mathbf{RH}_{\infty,0}^{2 \times 2}$ for each $\nu_i \in \mathcal{V}_B$, with at most one pole at $s=0$, and \lhc{$\widetilde{L}_H(s) \in \mathbf{RH}_{\infty,0}^{2N_B \times 2N_B}$}}.
    Let \cref{assump:connected,assump:low_freq_simp,assump:simple_eigenvalue} hold, \lhc{and suppose that \eqref{eq:low_freq_conditions} is satisfied at $s = 0$}. 
    Then $[\widetilde{G}_B(s),\widetilde{N}_H(s)]$ is stable under \cref{def:stability} if
    there exists an $\bar{\epsilon}>0$ such that for all $\epsilon \in (0,\bar{\epsilon}]$
    at least one of the following conditions is satisfied at each \lhc{$s = j\omega$, $\omega \geq \epsilon$}: 
    \begin{itemize}
        \item \lhc{for every $\tau \in [0,1]$}, there exists a multiplier $\widetilde{\Pi}(s) \in \mathbb{C}^{4N_B \times 4N_B}$ of the form \eqref{eq:multiplier_decentralised} such that \eqref{eq:iqc_decentralised_network} holds and, at each bus $\nu_i \in \mathcal{V}_B$, \eqref{eq:iqc_decentralised_buses} holds. 
        \item for each $(\nu_i,\nu_j) \in \widetilde{\mathcal{E}}_N$, there exists $\widetilde{\pi}_{12}^{ij}(s) \in \mathbb{C}^{2 \times 2}$ and $\widetilde{\pi}_{22}^{ij}(s) \in \mathbb{H}^2$  with $\widetilde{\pi}_{22}^{ij}(s) \geq 0$ such that \eqref{eq:iqc_ext_condition} is satisfied and, at each bus $\nu_i \in \mathcal{V}_B$, \eqref{eq:iqc_ext_local} holds. 
        \item $|s| \leq \omega_l$, \lhc{and conditions \eqref{eq:low_freq_conditions} hold}.
        \item condition \eqref{eq:loci_condition} holds.
    \end{itemize}
\end{proposition}
\begin{proof}
    \lhg{The proof is given in Appendix \ref{app:stability_proof}.}
\end{proof}

\lhc{\Cref{prop:stability} therefore certifies stability using conditions spanning} the $IV$ model at high frequencies, the $PQ$ model at low frequencies, and all three layers of locality (bus, extended-subsystem, and centralised). \lhc{The implementable conditions are those of [Part I, \cref{pt1-cor:pos_real,pt1-cor:dw_shell,pt1-cor:inf_norm,pt1-cor:fixed_net,pt1-cor:ext_stability_simp,pt1-cor:ext_stability,pt1-cor:conic}], transferred to the hybrid model via \cref{cor:iqc_decentralised_imped}, with the centralised condition \eqref{eq:loci_condition} covering any remaining regions of the contour.}

\section{Case Study: Kundur Two-Area System} \label{sec:case_study}
To test the validity of the proposed decentralised small-signal stability framework, we now demonstrate its application to the Kundur two-area system \cite[Example 12.6]{kundur_PowerSystem_22}. For particular control structures, this system is known to be stable, but this conclusion is typically reached through centralised analysis. 

The system consists of four generators ($G1$--$G4$) and eleven buses, as shown in \cref{fig:kundur}. Each generator was represented by the six-state model of \cite[Section 13.3.2]{kundur_PowerSystem_22} \lhc{with a thyristor exciter, a power system stabiliser and a speed-governing turbine. Lines, transformers, loads and shunts were modelled as in [Part I, \cref{pt1-sec:model_network}]. Network parameters and setpoints are those of \cite[Example 12.6]{kundur_PowerSystem_22}, and the controller transfer functions and gains are given in Appendix \ref{app:case_study_params}.}

\begin{figure}[t]
    \centering
    \includegraphics[width=0.48\textwidth]{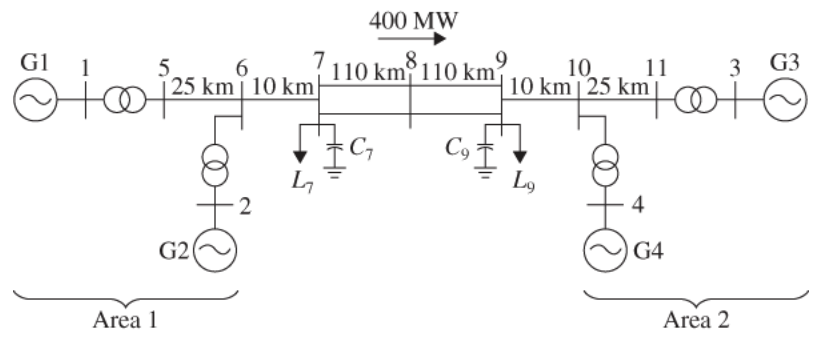}
    \caption{The Kundur two-area model. Adapted from \cite{kundur_PowerSystem_22}.}
    \label{fig:kundur}
\end{figure}

The system was linearised about its steady-state operating point \lhc{in the common $DQ$ reference frame}. 
\lhc{Each generator was combined with its transformer, together with the shunt branch of the connected line so that the resulting bus impedance was proper}, giving the impedances $Z_i(s)$ at the four buses $\mathcal{V}_B = \{\nu_5, \nu_6, \nu_{10}, \nu_{11}\}$. 
Finally, the \lhc{branches} between buses $6$ and $10$ were \lhc{merged into a single} admittance $Y_{6,10}(s)$ with shunt admittances $Y_6(s)$ and $Y_{10}(s)$. 

\subsection{Unstable Subsystems}
Once each $Z_i(s)$ was calculated, the system was converted to the $PQ$ model using \eqref{eq:N_def} and \eqref{eq:pq_from_iv}. Unstable poles were found in both $Z_i(s)$ and $G_i(s)$ at every $\nu_i \in \mathcal{V}_B$, \lhc{by the two mechanisms described in \cref{sec:model_instab}}.
\lhc{The equilibrium reactive power injections satisfied $Q_i^\star > 0$ at all four buses, so \cref{lemma:iv_unstable} predicts the low-frequency unstable poles observed in $Z_i(s)$.}
Those in $G_i(s)$ arose from encirclements over \SIrange{2.0e4}{2.1e4}{\radian\per\second}, \lhc{near the resonance of the series transformer branch, as predicted by \cref{lemma:pq_unstable}}.

\lhc{The two mechanisms occur at opposite ends of the frequency spectrum, so as anticipated in \cref{sec:model_hybrid} a single cutoff placed between them suffices.}
Applying the high-pass filter
\begin{equation}
    F_{H,i}(s) = \frac{s^3}{(s+\omega_{c,i})(s^2+2\zeta_{i}\omega_{c,i}s + \omega_{c,i}^2)} I_2
\end{equation}
with $\omega_{c,i}=\SI{15}{\radian\per\second}$ and $\zeta_i = 0.4$ at each bus gives $\widetilde{G}_i(s) \in \mathbf{RH}_{\infty,0}^{2 \times 2}$ via \eqref{eq:hybrid_bus_def} \lhg{(with one pole at $s=0$) as required}.
The hybrid network $\widetilde{N}_H(s)$ was then created using \eqref{eq:hybrid_net_def}.

\subsection{Local Stability Analysis}
\Cref{prop:stability} was then applied to the \lhi{system with the hybrid representation}, with results \lhc{for the decentralised conditions} shown in \cref{fig:case_study_result}. The conditions in the figure are split by locality layer, as specified in Part I. 
The bus-level conditions were tested using \cref{cor:iqc_decentralised_imped} and the multipliers from Part I. Specific conditions tested were the positive-real [Part I, \cref{pt1-cor:pos_real}], DW shell separation [Part I, \cref{pt1-cor:dw_shell}], infinity norm [Part I, \cref{pt1-cor:inf_norm}], small gain, small phase, rotated positive-real\footnote{It was verified that $(I_{N_B} \otimes J)Y_N(j\omega) - Y_N(j\omega)^\ast (I_{N_B} \otimes J)  \geq 0$ for $\omega \leq \omega_0$, with $\omega_0 = 2\pi \times \SI{60}{\hertz} = \SI{377}{\radian\per\second}$.} (PR) [Part I, \cref{pt1-cor:fixed_net}], and conic combination [Part I, \cref{pt1-cor:conic}] conditions. \lhh{The extended-subsystem conditions were tested using \cref{lemma:iqc_ext}}. \Cref{fig:num_range} also illustrates the graphical form of [Part I, \cref{pt1-cor:ext_stability_nr}] at two sample frequencies.

The positive-real condition held for all $\omega > \SI{1580}{\radian\per\second}$ and the extended-subsystem conditions for $\omega >\SI{3.28}{\radian\per\second}$
\lhc{(except for a gap over \SIrange{2160}{2282}{\radian\per\second} that the positive-real condition covers). Conditions for fixed networks were also feasible above \SI{11.27}{\radian\per\second}.
Hence, \eqref{eq:loci_condition} holds for $\omega >\SI{3.28}{\radian\per\second}$, which can be entirely determined using conditions that support plug-and-play operation --- unlike in Part~I, where the unstable $Z_B(s)$ prevented the application of extended-subsystem conditions}. 

\begin{figure}[t]
    \centering
    \includegraphics[width=0.48\textwidth]{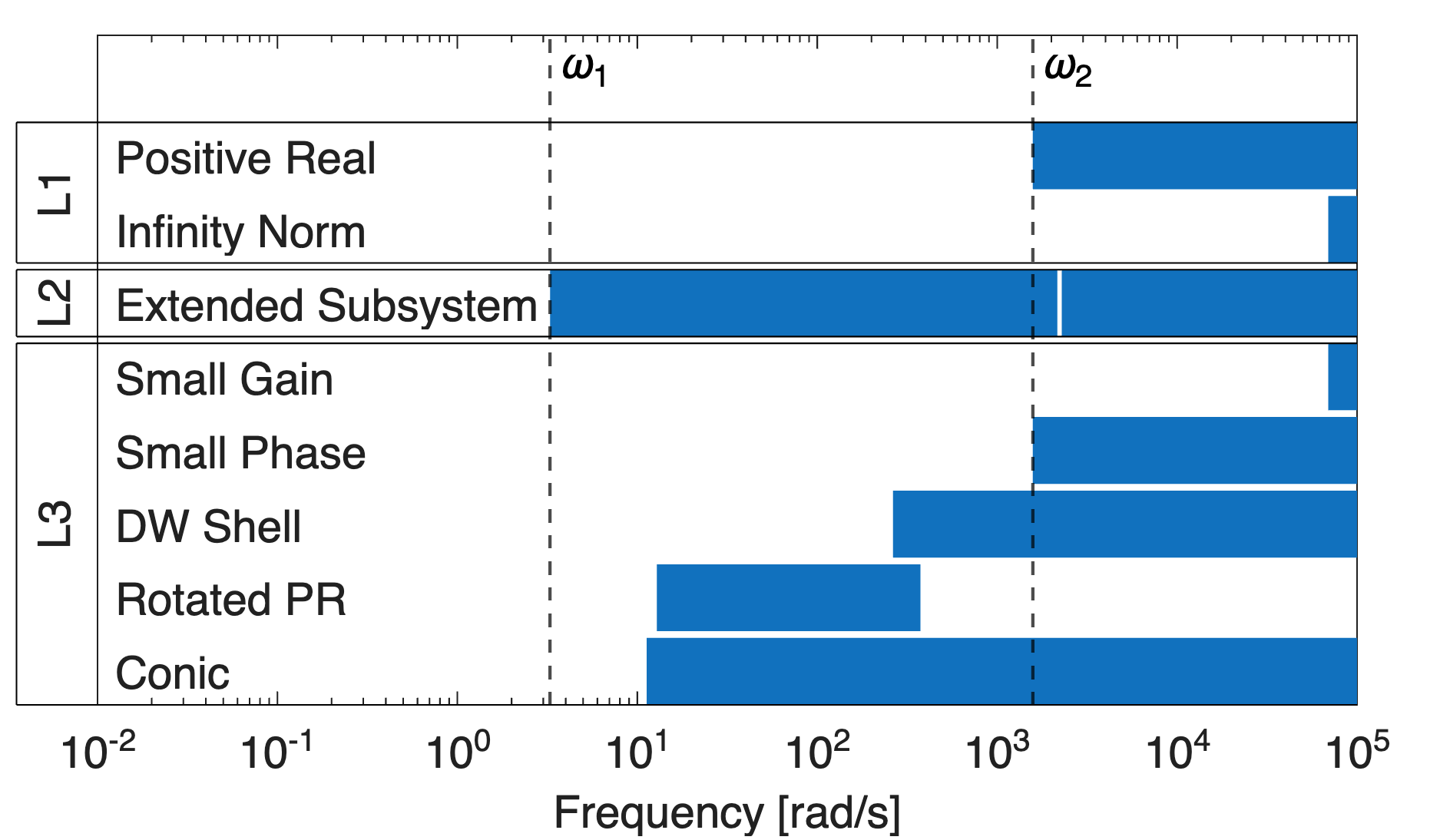}
    \caption{Frequencies over which [Part I, \cref{pt1-cor:pos_real,pt1-cor:dw_shell,pt1-cor:inf_norm,pt1-cor:fixed_net,pt1-cor:conic}] and \cref{lemma:iqc_ext} hold, split by layer of locality (decentralised bus (L1), decentralised extended subsystem (L2), and local conditions for fixed networks (L3)). Here, $\omega_1 = \SI{3.28}{\radian\per\second}$ and $\omega_2 = \SI{1580}{\radian\per\second}$.}
    \label{fig:case_study_result}
\end{figure}

\begin{figure}[t]
    \centering
    \begin{subfigure}{0.24\textwidth}
        \centering
        \includegraphics[width=0.95\textwidth]{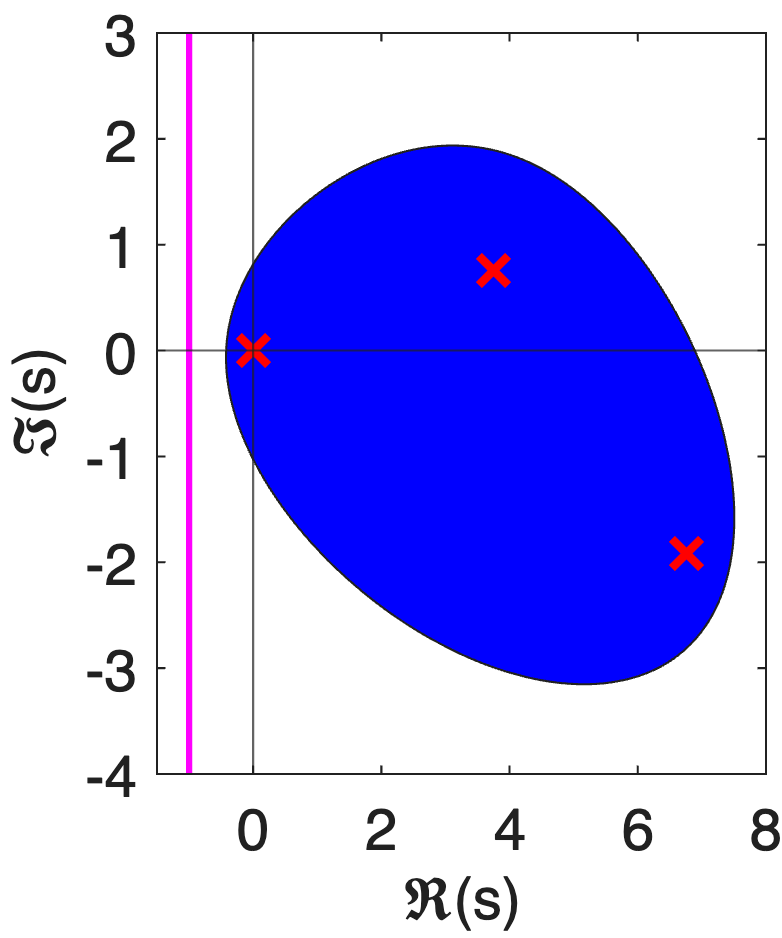}
        \caption{$\omega = \SI{30}{\radian\per\second}$.}
        \label{fig:num_range_30}
    \end{subfigure}
    \begin{subfigure}{0.24\textwidth}
        \centering
        \includegraphics[width=0.95\textwidth]{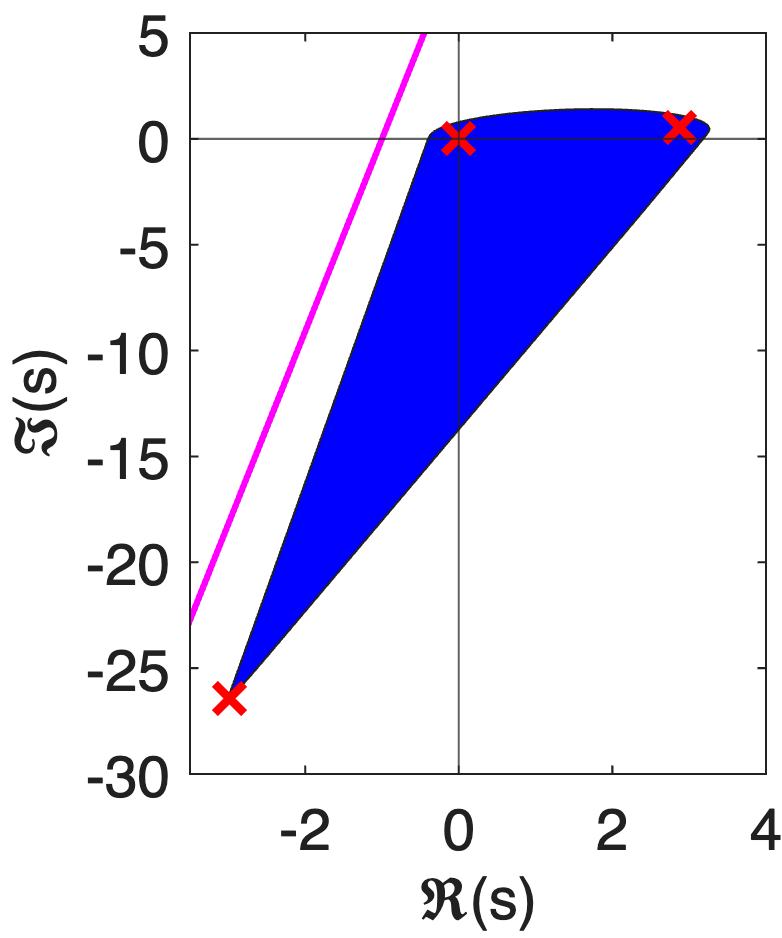}
        \caption{$\omega = \SI{380}{\radian\per\second}$.}
        \label{fig:num_range_380}
    \end{subfigure}
    \caption{Numerical ranges of $\widetilde{C}_{E,i}(s)$ (defined analogously to [Part I, \cref{pt1-cor:ext_stability_nr}]) for bus $\nu_6 \in \mathcal{V}_B$ at two sample frequencies. Eigenvalues are shown with red crosses. In both cases, the numerical range lies to the right of a line through $-1$ (in magenta) with $\bar{\eta}(30j) = 0$ and $\bar{\eta}(380j)=0.11$.}
    \label{fig:num_range}
\end{figure}

\subsection{Centralised Low-Frequency Analysis}
Below \SI{3.28}{\radian\per\second}, the decentralised conditions of \cref{lemma:iqc_decentralised,lemma:iqc_ext} were infeasible, so the centralised analysis of \cref{sec:low_freq} was used. 
\lhc{Letting $\omega_l =\SI{3.28}{\radian\per\second}$,  \cref{assump:low_freq_simp,assump:simple_eigenvalue} were both verified numerically, with $\widetilde{N}_H(j\omega) \approx \widetilde{N}_H(0)$ and $\mathrm{rank}(\widetilde{L}_{\omega\delta}(s)) = 3$ for $\omega < \omega_l$, and \eqref{eq:low_freq_conditions} was confirmed to hold at $s = 0$, as required by \cref{prop:stability}.}

\Cref{fig:ny_unbounded,fig:ny_bounded} show the unbounded and bounded characteristic loci of $\widetilde{L}_H(s)$ over the remainder of the contour, \lhc{where it can be seen} that none passes to the left of the point $-1$ on the real axis. Each figure compares the exact loci of $\widetilde{L}_H(s)$ (blue) with those obtained by fixing $\widetilde{N}_H(s) = \widetilde{N}_H(0)$ (red), and with the decoupled approximation of \cref{lemma:char_poly} (yellow). The blue and red curves overlap, confirming \cref{assump:low_freq_simp} over this range, \lhc{and both are closely tracked by the yellow curves, confirming that \eqref{eq:char_poly_decoupled} decouples as claimed}.

\begin{figure}[t]
    \centering
    \begin{subfigure}{0.48\textwidth}
        \centering
        \includegraphics[width=0.99\textwidth]{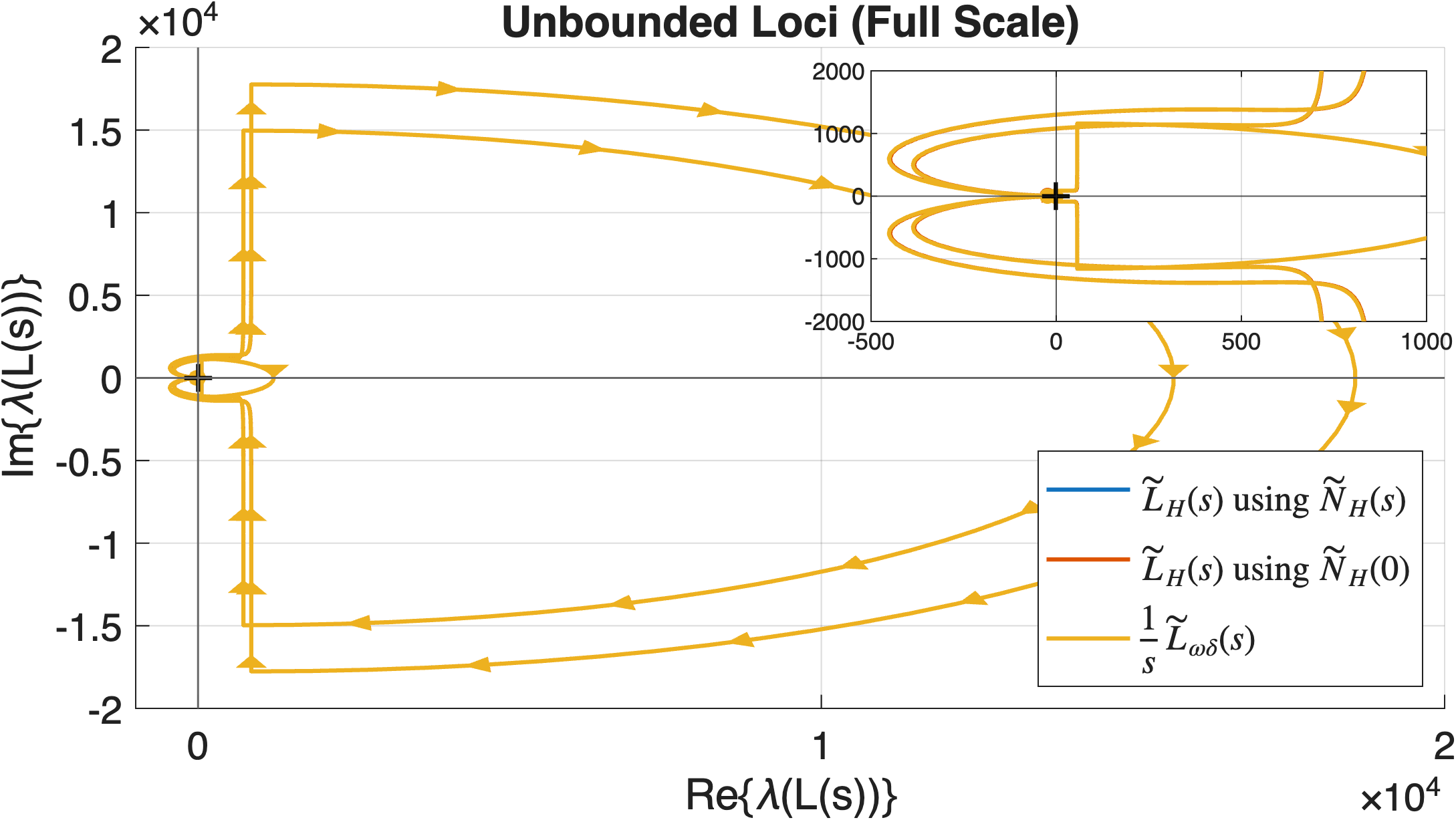}
        \caption{Unbounded loci of $\widetilde{L}_H(s)$.}
        \label{fig:ny_unbounded_full}
        \vspace*{1em}
    \end{subfigure}
    \\
    \begin{subfigure}{0.48\textwidth}
        \centering
        \includegraphics[width=0.99\textwidth]{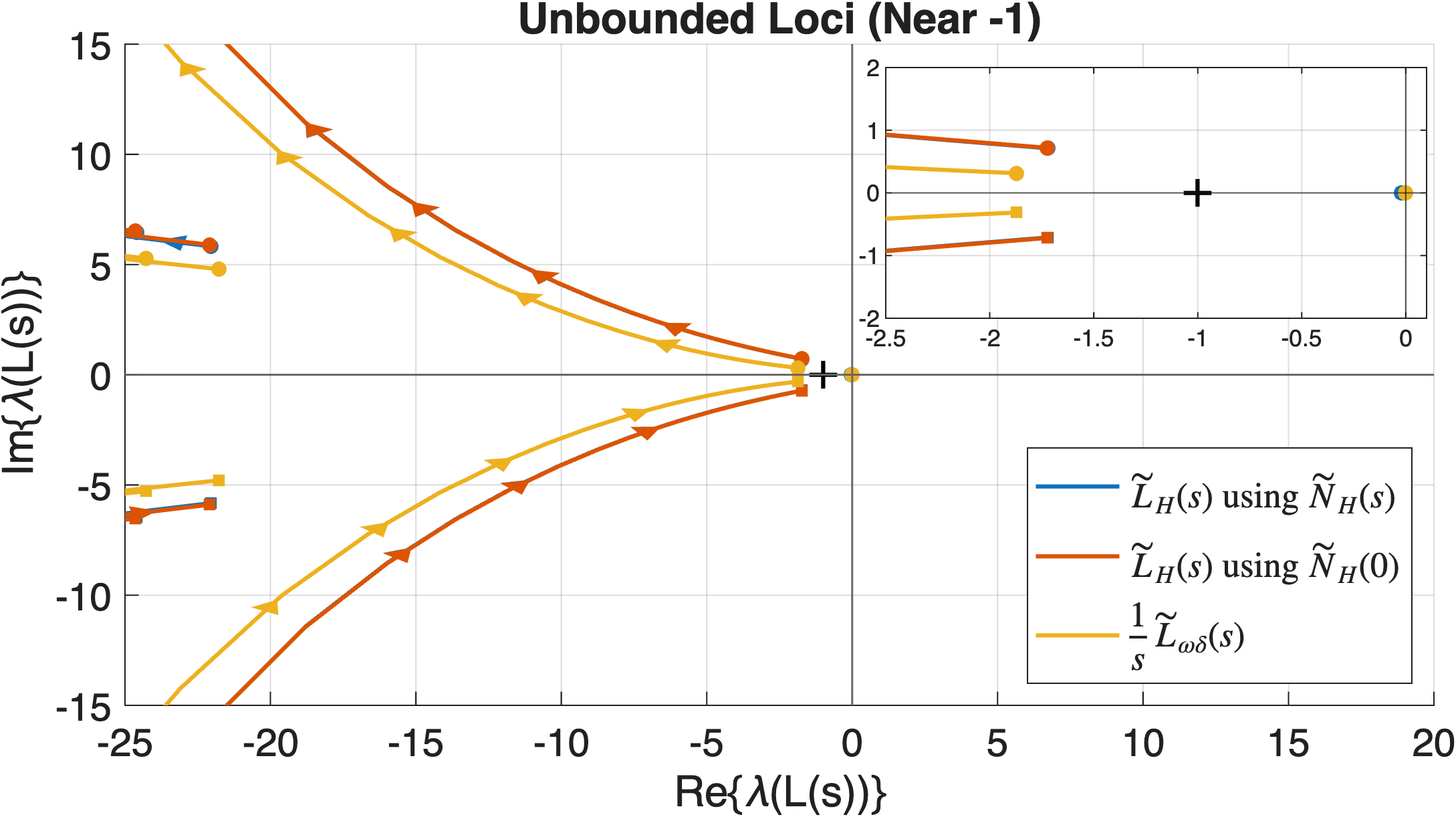}
        \caption{Unbounded loci of $\widetilde{L}_H(s)$, zoomed near $-1$.}
        \label{fig:ny_unbounded_zoomed}
    \end{subfigure}
    \caption{The Nyquist plot of $\widetilde{L}_H(s)$ showing only unbounded characteristic loci on the portion of the contour where $\omega \leq \SI{3.28}{\radian\per\second}$, including $\Gamma_\epsilon$, with $\epsilon = 1 \times 10^{-2}$. Insets show further detail.}
    \label{fig:ny_unbounded}
\end{figure}
\begin{figure}[t]
    \centering
    \includegraphics[width=0.48\textwidth]{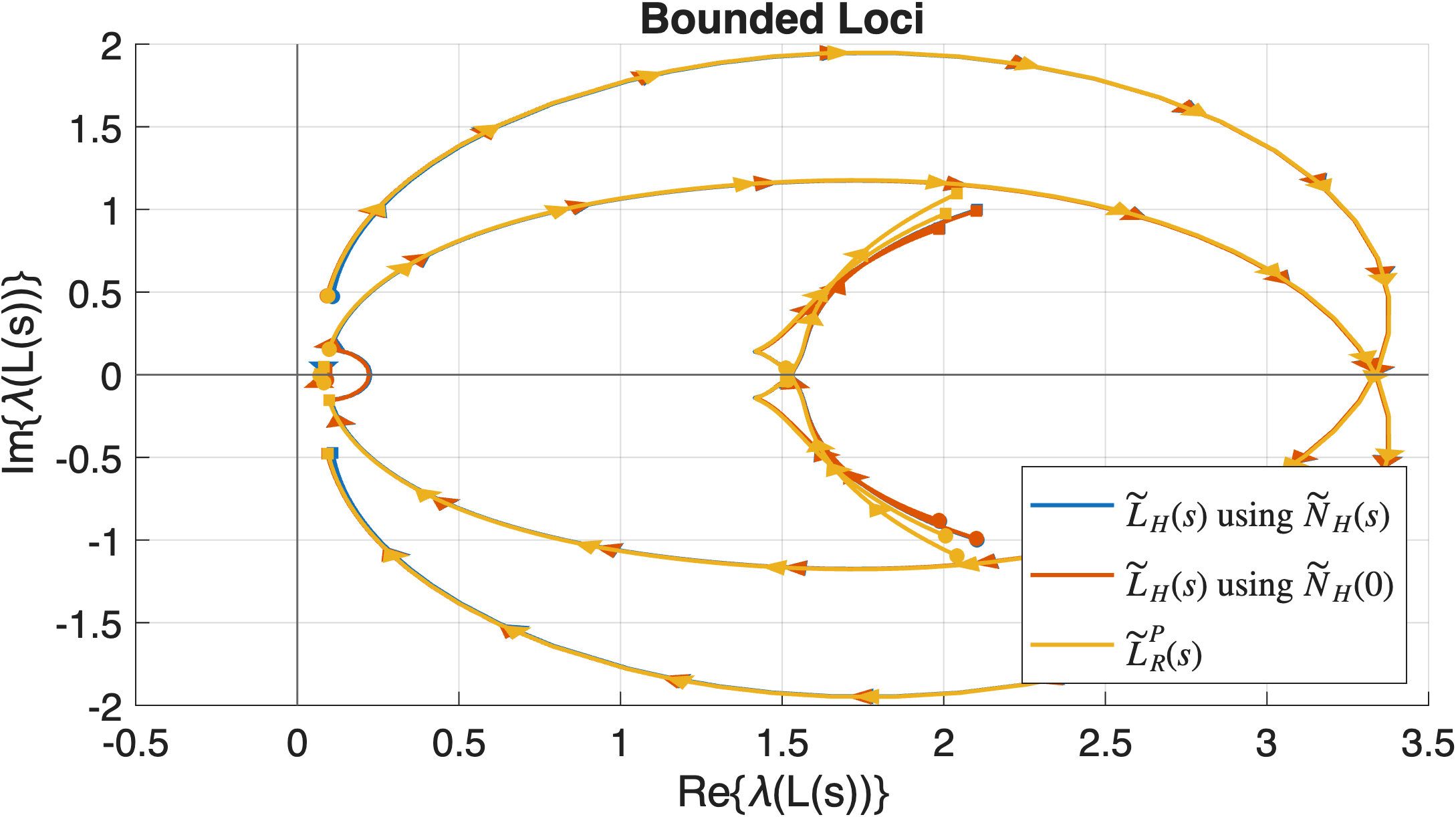}
    \caption{The Nyquist plot of $\widetilde{L}_H(s)$ showing only bounded characteristic loci on the portion of the contour where $\omega \leq \SI{3.28}{\radian\per\second}$, including $\Gamma_\epsilon$, with $\epsilon = 1 \times 10^{-2}$.}
    \label{fig:ny_bounded}
\end{figure}

The conditions of \cref{cor:low_freq} are therefore satisfied: the unbounded loci \lhc{diverge along the imaginary axis and sweep through the right half-plane on} $\Gamma_\epsilon$, while the bounded loci remain in the right half-plane.
\lhc{This is the behaviour anticipated in \cref{sec:low_freq_behaviour}, where the bounded
loci are governed by the reduced power-flow Jacobian (the eigenvalues of
$J_{qv} - J_{q\delta}J_{p\delta}^{\#}J_{pv}$ in \eqref{eq:reduced_jacobian} were
confirmed to lie in the right half-plane).}
\lhc{Provided operating points remain sensible, this low-frequency verification is an infrequent system-level check that need not be repeated every time a device connects.}

\lhc{Every condition of \cref{prop:stability} is therefore satisfied, certifying that $[\widetilde{G}_B(s),\widetilde{N}_H(s)]$ is stable. For higher frequencies, this conclusion was reached through decentralised analysis using conditions that lead to plug-and-play operation, providing a grid code for future modifications to this network.}

\section{Conclusion} \label{sec:conclusion}
Over two parts, this paper presented a decentralised framework for certifying small-signal stability in AC grids. \lhi{This is based on} quadratic constraints on the input--output properties of individual subsystems, \lhi{formulated as conditions on corresponding frequency response functions}. The framework makes no simplifications of the underlying dynamics, and, through the extensions of Part II, accommodates networks with unstable subsystems and without a stiff voltage source. 

Part II addressed two limitations identified in Part I. First, by linking the impedance ($IV$) and power $(PQ)$ representations, we showed analytically that unstable subsystem poles arise in many common scenarios, and that the two mechanisms are separated in timescale. A hybrid model exploiting that separation removes them, so that the stability framework of Part I applies directly and the extended-subsystem conditions that enable plug-and-play operation are recovered. 
Second, a direct characterisation of the low-frequency Nyquist plot showed that the characteristic loci decouple into bounded and unbounded branches along the classical $P$--$\delta$ / $Q$--$V$ separation. Sufficient conditions on each set --- when combined with a decentralised high-frequency analysis --- allow us to certify that the Nyquist criterion holds for the hybrid \lhi{representation}. These conditions relate to the reduced power-flow Jacobian already constrained during generator dispatch, so grid codes can concentrate on higher frequencies.

Both results were verified on the Kundur two-area system. Future work will address \lhg{grids containing} internal nodes with zero net current injection, practical selection of \lhi{specifications}, and further validation studies.

\appendix

\subsection{Proof of \cref{lemma:iv_unstable}} \label{app:iv_unstable_proof}
\begin{proof}
In what follows, we omit the subscript $(\cdot)_i$ and superscript $(\cdot)^\star$ in the system parameters to simplify the notation. As $s\rightarrow 0$, the Laurent expansion of $G_i(s)$ gives
\begin{equation*}
    G_i(s) = \begin{bmatrix}
        \frac{k_\omega}{s} - k_\omega \tau_\omega + k_\omega \tau_\omega^2 s& 0 \\ 0 & k_v-k_v\tau_v s
    \end{bmatrix} + \mathcal{O}(s^2).
\end{equation*}
The eigenvalues of $L_{Z,i}(s)$ are given by
\begin{equation} \label{eq:lemma_eig_formula}
    \begin{aligned}
        \lambda_\pm(s) = & \tfrac{1}{2} \left(\mathrm{tr}(L_{Z}(s)) \right. \\
        & \qquad \left. \pm \sqrt{\mathrm{tr}(L_{Z}(s))^2 - 4\det(L_{Z}(s))}\right).
    \end{aligned}
\end{equation}
We have
\begin{equation*}
    \begin{aligned}
        \mathrm{tr}(L_{Z,i}(s)) =
            &-Qk_\omega\frac{1}{s} + Q(k_\omega \tau_\omega + k_v) \\
            &\qquad- Q(k_\omega \tau_\omega^2+ k_v\tau_v)s  + \mathcal{O}(s^2), \\
        \mathrm{tr}(L_{Z,i}(s))^2 =
            & Q^2k_\omega^2\frac{1}{s^2} - 2Q^{2}k_\omega(k_\omega \tau_\omega + k_v)\frac{1}{s} \\
            &\qquad + 2Q^{2}k_\omega(k_\omega\tau_\omega^2 + k_v\tau_v) \\
            &\qquad + Q^2(k_\omega \tau_\omega + k_v)^2 + \mathcal{O}(s), \\
                                \det(L_{Z,i}(s)) =
            & - \left(Q^2 + P^2\right)\left( k_\omega k_v \frac{1}{s} \right. \\
            & \qquad \left.- k_\omega k_v(\tau_\omega + \tau_v) \right) + \mathcal{O}(s).
    \end{aligned}
\end{equation*}
Therefore,
\begin{equation*}
    \begin{aligned}
        \mathrm{tr}&(L_{Z,i}(s))^2  - 4\det(L_{Z,i}(s)) \\
        & = Q^2 k_\omega^2 \frac{1}{s^2} - 2 Q^2 k_\omega(k_\omega\tau_\omega - k_v)\frac{1}{s} + 4P^2 k_\omega k_v \frac{1}{s} \\
            &\qquad  + Q^2(k_\omega \tau_\omega - k_v)^2 +2Q^2 k_\omega(k_\omega\lhc{\tau_\omega^2} - k_v\tau_v) \\
            & \qquad -4P^2k_\omega k_v(\tau_\omega+\tau_v) + \mathcal{O}(s) \\
        & = Q^2\left(k_\omega \frac{1}{s} - (k_\omega \tau_\omega - k_v) \right)^2 + 4P^2k_\omega k_v \frac{1}{s} + \mathcal{O}(1).
    \end{aligned}
\end{equation*}
Let
$
    h(s) = Q\left(k_\omega \frac{1}{s} - (k_\omega \tau_\omega - k_v) \right)
$.
Then
\begin{equation*}
    \begin{aligned}
        \mathrm{tr}(L_{Z,i}&(s))^2  - 4\det(L_{Z,i}(s)) \\
        & = h(s)^2 + 4P^2k_\omega k_v \frac{1}{s} + \mathcal{O}(1) \\
        & = h(s)^2 \left[1 + h(s)^{-2} \left(4P^2k_\omega k_v \frac{1}{s} + \mathcal{O}(\lhc{1}) \right) \right].
    \end{aligned}
\end{equation*}
Now
$
    h(s)^{-2} = \frac{s^2}{Q^2 k_\omega^2} + \mathcal{O}(s^3)
$,
so
\begin{equation*}
    \begin{aligned}
        \mathrm{tr}(L_{Z,i}(s))^2  &- 4\det(L_{Z,i}(s)) \\
        & = h(s)^2 \left[1 + \frac{4P^2 k_v}{Q^2 k_\omega}s + \mathcal{O}(s^2) \right].
    \end{aligned}
\end{equation*}
Therefore
\begin{equation*}
    \begin{aligned}
        & \sqrt{ \mathrm{tr}(L_{Z,i}(s))^2  - 4\det(L_{Z,i}(s)) } \\
        & \qquad = h(s) \left[1 + \frac{2P^2 k_v}{Q^2 k_\omega}s + \mathcal{O}(s^2) \right] \\
        & \qquad = Qk_\omega \frac{1}{s} - Q k_\omega\tau_\omega + Q k_v + \frac{2P^2 k_v}{Q} + \mathcal{O}(s).
    \end{aligned}
\end{equation*}
Now, evaluating $\lambda_\pm(s)$ using \eqref{eq:lemma_eig_formula} gives \eqref{eq:instab_iv_loci}.
\end{proof}

\subsection{Proof of \cref{lemma:pq_unstable}} \label{app:pq_unstable_proof}
\begin{proof}
By direct calculation, we get
\begin{equation*}
    \check{Z}_i^p(s) = \frac{1}{ {V_i^\star}^2 } \begin{bmatrix}
        -b(s) & -a(s) \\ a(s) & -b(s)
    \end{bmatrix}.
\end{equation*}
Therefore
\begin{equation*}
    L_{G,i}(s) = \frac{-1}{ {V_i^\star}^2 } \begin{bmatrix}
        - a(s)P_i^\star +b(s)Q_i^\star  & -a(s)Q_i^\star-b(s)P_i^\star  \\ -a(s)Q_i^\star - b(s)P_i^\star & a(s)P_i^\star -b(s)Q_i^\star
    \end{bmatrix}.
\end{equation*}
The eigenvalues of $L_{G,i}(s)$ are given by the solutions to $\det(\lambda I-L_{G,i}(s)) = 0$. This gives
\begin{equation*}
    \begin{aligned}
        &\lambda^2 \!-\! \frac{1}{ {V_i^\star}^4 } \left[(a(s)P_i^\star \!-\! b(s)Q_i^\star)^2 \!+\! (a(s)Q_i^\star \!+\! b(s)P_i^\star)^2 \right] = 0 \\
        &\implies \lambda^2 - \frac{{S_i^\star}^2}{ {V_i^\star}^4 } \left[ a(s)^2 + b(s)^2 \right] = 0.
    \end{aligned}
\end{equation*}
Solving the above gives \eqref{eq:instab_pq_loci}.
\end{proof}

\subsection{Proof of \cref{cor:iqc_decentralised_imped}} \label{app:iqc_decentralised_imped_proof}
\begin{proof}
In what follows, we omit the argument $s$ to simplify the notation. As \eqref{eq:iqc_decentralised_network_Y} holds, we have by [Part I, \cref{pt1-cor:iqc_convex}] that 
\begin{equation*}
    \tau^2 Y_N^\ast \Pi_{22} Y_N - \tau \Pi_{12}^\ast Y_N - 
    \tau Y_N^\ast \Pi_{12} + \Pi_{11} \leq 0,
\end{equation*}
for all $\tau \in [0,1]$. Now, defining $V^\star:= (\oplus_{\nu_i \in \mathcal{V}_B} V_i^\star)\otimes I_2$ and using \eqref{eq:U_squared}, we can write
\begin{equation*}
    \begin{aligned}
        & \tau^2 {V^\star}^{-2} U^\ddagger U^\ddagger Y_N^\ast U^\dagger U^\dagger {V^\star}^{-2}  \Pi_{22}  {V^\star}^{-2} U^\dagger U^\dagger Y_N  U^\ddagger U^\ddagger {V^\star}^{-2} \\
        & - \tau \Pi_{12}^\ast {V^\star}^{-2} U^\dagger U^\dagger Y_N U^\ddagger U^\ddagger {V^\star}^{-2} \\ 
        & - \tau {V^\star}^{-2} U^\ddagger U^\ddagger Y_N^\ast  U^\dagger U^\dagger {V^\star}^{-2} \Pi_{12}
        + \Pi_{11} \leq 0.
    \end{aligned}
\end{equation*}
Definitions \eqref{eq:hybrid_net_def} and \eqref{eq:W_HL_def} give $U^\dagger Y_N U^\ddagger = \widetilde{N}_H-W_L$ and pre- and post-multiplying by $U^\ddagger$ gives 
\begin{equation*}
    \begin{aligned}
        & \tau^2 (\widetilde{N}_H^\ast- W_L^\ast) U^\dagger {V^\star}^{-2} \Pi_{22} {V^\star}^{-2} U^\dagger (\widetilde{N}_H-W_L) \\
        & - \tau {V^\star}^{-2} U^\ddagger \Pi_{12}^\ast U^\dagger (\widetilde{N}_H-W_L) \\  
        & - \tau (\widetilde{N}_H^\ast- W_L^\ast) U^\dagger \Pi_{12} U^\ddagger {V^\star}^{-2}
        + U^\ddagger \Pi_{11} U^\ddagger \leq 0.
    \end{aligned}
\end{equation*}
Using $U^\dagger W_L = {V^\star}^2 \check{W}_L U^\ddagger$ (from \eqref{eq:W_HL_def}), this can be written in the form \eqref{eq:iqc_decentralised_network} where 
\begin{equation} \label{eq:mult_blocks_transformed}
    \begin{aligned}
        &\begin{aligned}
            \widetilde{\Pi}_{11}^i = & \ U_i^\ddagger \left[ \Pi_{11}^i + \tau  \check{W}_{L,i}^\ast \Pi_{12}^i  + \tau  {\Pi_{12}^i}^\ast  \check{W}_{L,i} \right.  \\
            & \left. + \tau^2  \check{W}_{L,i}^\ast   \Pi_{22}^i \check{W}_{L,i} \right] U_i^\ddagger,
        \end{aligned} \\
        &\widetilde{\Pi}_{12}^i = {V_i^\star}^{-2} U_i^\dagger \left[ \Pi_{12}^i  + \tau   \Pi_{22}^i \check{W}_{L,i} \right] U_i^\ddagger ,\\
        &\widetilde{\Pi}_{22}^i = {V_i^\star}^{-4} U_i^\dagger  \Pi_{22}^i  U_i^\dagger.
    \end{aligned}
\end{equation}

Now, by \cref{lemma:iqc_decentralised}, \eqref{eq:loci_condition} holds if, for each $\tau \in [0,1]$, \eqref{eq:iqc_decentralised_buses} holds at every $\nu_i \in \mathcal{V}_B$ with the multiplier blocks above. Substituting these blocks into the left-hand side of \eqref{eq:iqc_decentralised_buses}, pre- and post-multiplying by $U_i^\dagger$ and using \eqref{eq:hybrid_impd_def} gives
\begin{equation*}
\begin{aligned}
     g_i(\tau)  
     :=&\, \widetilde{Z}_i^\ast \Pi_{11}^i \widetilde{Z}_i + (I_2 + \tau \check{W}_{L,i}\widetilde{Z}_i)^\ast \Pi_{12}^i \widetilde{Z}_i \\
        &\quad  + \widetilde{Z}_i^\ast {\Pi_{12}^i}^\ast (I_2 + \tau \check{W}_{L,i}\widetilde{Z}_i) \\
        &\quad + (I_2 + \tau \check{W}_{L,i}\widetilde{Z}_i)^\ast \Pi_{22}^i (I_2 + \tau \check{W}_{L,i}\widetilde{Z}_i) \\
    = &\, \begin{bmatrix}
        \widetilde{Z}_i \\ I + \tau \check{W}_{L,i} \widetilde{Z}_i
    \end{bmatrix}^\ast \begin{bmatrix}
        \Pi_{11}^i & {\Pi_{12}^i}^\ast \\ \Pi_{12}^i & \Pi_{22}^i
    \end{bmatrix} \begin{bmatrix}
        \widetilde{Z}_i \\ I + \tau \check{W}_{L,i} \widetilde{Z}_i
    \end{bmatrix}.
        \end{aligned} 
\end{equation*}

The matrix polynomial $g_i(\tau)$ is quadratic in $\tau$, and we require $g_i(\tau) > 0 $ for all $\tau \in [0,1]$. By collecting powers of $\tau$, we obtain the identity
\begin{equation} \label{eq:chord}
    g_i(\tau) = (1-\tau) g_i(0) + \tau g_i(1) - \tau(1-\tau) A_i,
\end{equation}
where $A_i = \widetilde{Z}_i^\ast \check{W}_{L,i}^\ast \Pi_{22}^i \check{W}_{L,i} \widetilde{Z}_i$. Completing the square in the final term gives 
\begin{equation*} 
    \begin{aligned}
        g_i(\tau) 
        &= (1-\tau) g_i(0) + \tau g_i(1) - \left[\tfrac{1}{4}-(\tau - \tfrac{1}{2})^2 \right] A_i \\
        &= (1-\tau) \left[g_i(0) - \tfrac{1}{4}A_i \right] + \tau \left[g_i(1) - \tfrac{1}{4}A_i\right]  \\
            &\qquad + (\tau - \tfrac{1}{2})^2 A_i.
    \end{aligned}
\end{equation*}
As $\Pi_{22}^i \geq 0$, we have $A_i \geq 0$. Therefore, $g_i(\tau) > 0$ for all $\tau \in [0,1]$ if both $g_i(0) - \tfrac{1}{4}A_i > 0$ and $g_i(1) - \tfrac{1}{4}A_i > 0$. 

The first of these constraints gives
\begin{equation*}
    \widetilde{Z}_i^\ast \left(\Pi_{11}^i - M_i\right)\widetilde{Z}_i
    + \Pi_{12}^i \widetilde{Z}_i + \widetilde{Z}_i^\ast {\Pi_{12}^i}^\ast
    + \Pi_{22}^i > 0,
\end{equation*}
which is \eqref{eq:iqc_decentralised_buses_Z_tilde} with $M_i$ in \eqref{eq:M_def}. Similarly, the second constraint gives 
\begin{equation} \label{eq:lemma_Y_Z_deriv}
\setlength{\arraycolsep}{0.1cm}
    \begin{bmatrix}
        \widetilde{Z}_i \\ I +  \check{W}_{L,i} \widetilde{Z}_i
    \end{bmatrix}^\ast \begin{bmatrix}
        \Pi_{11}^i \!-\! M_i & \!\! {\Pi_{12}^i}^\ast \\ \Pi_{12}^i & \!\! \Pi_{22}^i
    \end{bmatrix} \begin{bmatrix}
        \widetilde{Z}_i \\ I + \check{W}_{L,i} \widetilde{Z}_i
    \end{bmatrix} > 0.
\end{equation}
Now, factoring the outer matrix gives
\begin{equation*}
    \begin{bmatrix}
        \widetilde{Z}_i \\ I + \check{W}_{L,i} \widetilde{Z}_i
    \end{bmatrix} = \begin{bmatrix}
        \widetilde{Z}_i(I + \check{W}_{L,i} \widetilde{Z}_i)^{-1} \\ I
    \end{bmatrix} (I + \check{W}_{L,i} \widetilde{Z}_i),
\end{equation*}
and, using \eqref{eq:U_squared}, \eqref{eq:hybrid_impd_def} and \eqref{eq:W_HL_def}, the top block simplifies to
\begin{equation*}
    \begin{aligned}
        \widetilde{Z}_i(I + \check{W}_{L,i} \widetilde{Z}_i &)^{-1} 
         = U_i^\ddagger \widetilde{G}_i U_i^\dagger ( I + {V_i^\star}^{-2} U_i^\dagger W_{L,i} \widetilde{G}_i U_i^\dagger)^{-1} \\
        & = U_i^\ddagger \widetilde{G}_i (I+ W_{L,i} \widetilde{G}_i)^{-1}  U_i^\dagger \\
        & = U_i^\ddagger G_i (I+ W_{H,i} G_i)^{-1} (I+ W_{L,i} \widetilde{G}_i)^{-1}  U_i^\dagger,
    \end{aligned}
\end{equation*}
where in the last equality, we used \eqref{eq:hybrid_bus_def} and the push-through identity $(I+G_i W_{H,i})^{-1} G_i = G_i (I + W_{H,i} G_i)^{-1}$. Then
\begin{equation*}
    \begin{aligned}
        (I&+W_{H,i}  G_i)^{-1} ( I + W_{L,i} \widetilde{G}_i)^{-1} \\
        & = [ ( I + W_{L,i} \widetilde{G}_i)(I+W_{H,i}G_i)]^{-1} \\
        & = [ ( I + W_{L,i} G_i(I+W_{H,i}G_i)^{-1})(I+W_{H,i}G_i)]^{-1} \\
        & = [I+W_{H,i}G_i + W_{L,i} G_i ]^{-1} \\
        & = [I + W_i G_i]^{-1},
    \end{aligned}
\end{equation*}
using $W_{H,i} + W_{L,i} = W_i$, so 
\begin{equation*}
    \widetilde{Z}_i(I + \check{W}_{L,i} \widetilde{Z}_i)^{-1} = U_i^\ddagger G_i (I+ W_i G_i)^{-1}  U_i^\dagger = Z_i
\end{equation*}
by \eqref{eq:iv_from_pq}. Therefore, \eqref{eq:lemma_Y_Z_deriv} holds if and only if \eqref{eq:iqc_decentralised_buses_Z_no_tilde} holds by congruence. 
\end{proof}

\subsection{Proof of \cref{lemma:char_poly}} \label{app:char_poly_proof}
\begin{proof}
The characteristic polynomial is given by 
\begin{equation*}
    \det \begin{bmatrix}
    \frac{1}{s} \widetilde{L}_{\omega \delta}(s) -\lambda I & \frac{1}{s} \widetilde{L}_{\omega v}(s) \\ \widetilde{L}_{v \delta}(s) & \widetilde{L}_{vv}(s) - \lambda I
    \end{bmatrix} = 0.
\end{equation*}
Using the Schur complement gives $\det(\frac{1}{s}\widetilde{L}_{\omega \delta}(s) -\lambda I)\det(\widehat{L}_R^P(s,\lambda) - \lambda I) = 0$, where
\begin{equation} \label{eq:open_loop_reduced}
    \widehat{L}_R^P(s,\lambda) = \widetilde{L}_{vv}(s) - \widetilde{L}_{v\delta}(s) (\widetilde{L}_{\omega\delta}(s) - s\lambda I)^{-1} \widetilde{L}_{\omega v}(s).
\end{equation}
Define the projector $P_0$ for $\widetilde{L}_{\omega\delta}(s)$ as in \eqref{eq:group_inv_proj}. Under \cref{assump:simple_eigenvalue}, $P_0$ takes the form
\begin{equation} \label{eq:P0_form}
    P_0 = \mathbf{1} \xi^T, \qquad \text{where} \quad 
    \xi^T \widetilde{L}_{\omega\delta}(s) = 0, \quad \xi^T \mathbf{1} = 1,
\end{equation}
i.e., $\xi$ is the normalised left null vector of $\widetilde{L}_{\omega\delta}(s)$. Hence $P_0$ is rank one and every column is proportional 
to $\mathbf{1}$. Now, Theorem 2.1 of \cite{rose_LaurentExpansion_78} applied to $(\widetilde{L}_{\omega\delta}(s)-s\lambda I)^{-1}$ near $s = 0$ gives 
\begin{equation*}
    \begin{aligned}
        (\widetilde{L}_{\omega\delta}&(s)-s\lambda I)^{-1} \\
        & = \widetilde{L}_{\omega\delta}(s)^\# 
        - (I - \widetilde{L}_{\omega\delta}(s) \widetilde{L}_{\omega\delta}(s)^\#) \frac{1}{s\lambda} + \mathcal{O}(s\lambda) \\
        & = \widetilde{L}_{\omega\delta}(s)^\#  - \frac{1}{s\lambda} \mathbf{1}\xi^T + \mathcal{O}(s\lambda),
    \end{aligned}
\end{equation*}
where we used \eqref{eq:group_invser_proj_prop} and \eqref{eq:P0_form}. Substitution into \eqref{eq:open_loop_reduced} near $s = 0$ and using \eqref{eq:null_identities} gives \eqref{eq:char_poly_decoupled}.
\end{proof}

\subsection{Proof of \cref{cor:low_freq}} \label{app:low_freq_proof}
\begin{proof}
    \lhc{Condition \eqref{eq:loci_condition} requires each branch of the characteristic loci to avoid $(-\infty, -1]$. The zero branch clearly satisfies this, and the bounded branch satisfies it by \eqref{eq:bounded_cond}. For the unbounded branches, let \eqref{eq:unbounded_cond} hold. Along $\Gamma_{j\omega}^{\epsilon-}$, $\Im(\lambda) > 0$ and along $\Gamma_{j\omega}^{\epsilon+}$, $\Im(\lambda) < 0$. Along $\Gamma_\epsilon$, $\lambda$ traces a semicircle of radius $\frac{|\mu|}{\epsilon}$ whose argument decreases from $\arg\mu + \tfrac{\pi}{2}$ to $\arg\mu - \tfrac{\pi}{2}$. Since $\Re(\mu) > 0$ gives $\lvert\arg\mu\rvert < \tfrac{\pi}{2}$, the argument never reaches $\pm\pi$. Hence $\lambda \notin (-\infty,-1]$ throughout.}
\end{proof}

\subsection{Proof of \cref{prop:stability}} \label{app:stability_proof}
\begin{proof}
\lhc{By \cref{lemma:iqc_decentralised,lemma:iqc_ext,cor:low_freq}, each condition gives \eqref{eq:loci_condition} on $\Gamma_{j\omega}^{\epsilon+}$. As in [Part I, \cref{pt1-prop:stability}], we can ignore $\Gamma_{j\omega}^{\epsilon-}$ and $\Gamma_R$.
For $\Gamma_\epsilon$, note that the non-zero eigenvalues of $\widetilde{L}_{\omega\delta}(s)$ and the eigenvalues of $\widetilde{L}_R^P(s)$ vary continuously with $s$ (the rank of $\widetilde{L}_{\omega\delta}(s)$ is constant by \cref{assump:simple_eigenvalue}, so $\widetilde{L}_{\omega\delta}(s)^\#$ is continuous). 
As \lhi{the conditions} \eqref{eq:low_freq_conditions} hold at $s = 0$, 
there \lhi{hence} exists $\bar{\epsilon} > 0$ such that they hold for all $\lvert s \rvert \leq \bar{\epsilon}$. Therefore, \cref{cor:low_freq} gives \eqref{eq:loci_condition} on $\Gamma_\epsilon$ for every $\epsilon \in (0,\bar{\epsilon}]$. 
\lhg{Finally, let $A(s) = \mathrm{diag}(sI,I) (I+\widetilde{L}_H^P(s))$, with $\widetilde{L}_H^P(s)$ in \eqref{eq:open_loop_permuted}. $A(s)$ is analytic near $s=0$ and, using a similar approach as in the proof of \cref{lemma:char_poly}, we get $\det(A(s)) = \det(sI+\widetilde{L}_{\omega\delta}(s)) \det(I + \widetilde{L}_R^P(s) + \mathcal{O}(s))$. The first determinant has a simple zero at $s=0$ by \cref{assump:simple_eigenvalue} and \eqref{eq:unbounded_cond}, and the second is non-zero there by \eqref{eq:bounded_cond}. Hence $\det(A(s))$ has a simple zero at the origin and $(I+\widetilde{L}_H(s))^{-1}$ has at most one pole there, meaning}
stability follows from [Part I, \cref{pt1-lemma:nyquist_sufficient}].}
\end{proof}

\subsection{Case Study Controllers} \label{app:case_study_params}
\textbf{Speed-governing turbine:}
\begin{equation*}
    \Delta P_m(s) = \frac{1}{R} \frac{1}{1 + sT_{a}} \frac{1+sT_{b}}{1 + sT_{c}} \Delta \omega(s)
\end{equation*}
where $\Delta P_m(s)$ is the mechanical power and $\Delta \omega(s)$ is the frequency, with time constants $T_{a} = \SI{0.5}{s}$, $T_{b}= \SI{2.1}{s}$ and $T_{c} = \SI{7}{s}$, and droop constant $R = \{0.04,0.05,0.04,0.06 \}$ for generators $\{G1, G2, G3, G4\}$. 

\textbf{Power system stabiliser (PSS):}
\begin{equation*}
    \Delta v_{PSS}(s) = K_{PSS} \frac{sT_W}{1+sT_W} \frac{1+sT_1}{1+sT_2}  \frac{1+sT_3}{1+sT_4} \Delta \omega(s), 
\end{equation*}
where $\Delta v_{PSS}(s)$ is the PSS output, with $K_{PSS} = 20$, $T_W = \SI{10}{s}$, $T_1 = \SI{0.05}{s}$, $T_2 = \SI{0.02}{s}$, $T_3 = \SI{3}{s}$, and $T_4 = \SI{5.4}{s}$.

\textbf{Thyristor exciter:}
\begin{equation*}
    \Delta E_{fd}(s) = K_A  \left( \frac{1}{1+sT_R} \Delta E_t(s) + \Delta v_{PSS}(s) \right) 
\end{equation*}
where $\Delta E_{fd}(s)$ is the field voltage and $\Delta E_t(s)$ is the terminal voltage magnitude, with $K_A = 200$ and $T_R = \SI{0.01}{s}$.

\section*{References}
\bibliographystyle{IEEEtran}
\bibliography{./Files/references.bib}
\end{document}